\documentclass[journal]{IEEEtran}

\usepackage{graphicx}
\usepackage{cite}
\usepackage[cmex10]{amsmath}
\usepackage{amsfonts}
\usepackage{amssymb}
\usepackage{subfigure}
\usepackage{color}
\usepackage{amsthm}
\usepackage{relsize}
\usepackage{epstopdf}
\usepackage{flushend}

\newcommand{\bm}{\mathbf} 
\newcommand{\be}{\begin{equation}}
\newcommand{\ee}{\end{equation}}
\newcommand{\bse}{\begin{subequations}}
\newcommand{\ese}{\end{subequations}}
\newcommand{\bea}{\begin{eqnarray}}
\newcommand{\eea}{\end{eqnarray}}
\newcommand{\x}{{\bm x}}

\newcommand{\q}{{\bm q}}
\newcommand{\ba}{{\bm a}}
\newcommand{\bb}{{\bm b}}
\newcommand{\bc}{{\bm c}}
\newcommand{\br}{{\bm r}}

\newcommand{\bA}{{\bm A}}

\newcommand{\bF}{{\bf F}}
\newcommand{\bD}{{\bf D}}

\newcommand{\bB}{{\bf B}}

\newcommand{\bG}{{\bf G}}
\newcommand{\bH}{{\bf H}}

\newcommand{\bg}{{\bf g}}

\newcommand{\bh}{{\bf h}}

\newcommand{\bd}{{\bf d}}

\newcommand{\bff}{{\bf f}}
\newcommand{\bzero}{{\bf 0}}

\newcommand{\eye}{{\bm I}}
\newcommand{\I}{{\bm I }}

\newcommand{\BG}{{\boldsymbol{\mathcal G}}}

\newcommand{\BW}{{\boldsymbol{\mathcal W}}}

\newcommand{\bdelta}{\mbox{\boldmath$\delta$}}

\newcommand{\bGamma}{\mbox{\boldmath$\Gamma$}}

\newcommand{\bPsi}{\mbox{\boldmath$\Psi$}}

\newcommand{\bnu}{\mbox{\boldmath$\nu$}}

\newcommand{\bxi}{\mbox{\boldmath{$\xi$}}}

\newcommand{\bPhi}{\mbox{\boldmath{$\Phi$}}}

\newcommand{\brho}{\mbox{\boldmath{$\rho$}}}
\newcommand{\trace}{\mathrm{tr}}
\newcommand{\sinr}{\mathrm{SINR}}
\newcommand{\mrc}{\mathtt{MRC}}
\newcommand{\zf}{\mathtt{ZF}}
\newcommand{\trzf}{\mathtt{TR\text{-}ZF}}
\newcommand{\tr}{\mathtt{TR}}
\newcommand{\trmrc}{\mathtt{TR\text{-}MRC}}
\newcommand{\cpofdm}{\mathtt{CP\text{-}OFDM}}

\newcommand{\Toeplitz}{\mathlarger{\mathcal{T}}}
\theoremstyle{definition}
\newtheorem{proposition}{Proposition}
\newtheorem{remark}{Remark}

\newcommand{\xbar}[1]{\mkern 3.5mu\overline{\mkern-3.5mu#1\mkern-3.5mu}\mkern 3.5mu}

\graphicspath{{figures/}} 

\title{OFDM Without CP in Massive MIMO}

\author{\normalsize Amir Aminjavaheri, Arman Farhang, Ahmad RezazadehReyhani,  Linda E. Doyle, and Behrouz Farhang-Boroujeny 

\thanks{Parts of the concepts based on which the contents of this paper are built have been presented in \cite{OFDMnoCPconf}.}

\thanks{A.~Aminjavaheri, A.~RezazadehReyhani and B.~Farhang-Boroujeny are with the Electrical and Computer Engineering Department, University of Utah, Salt Lake City, USA (e-mail: \{aminjav, rezazade, farhang\}@ece.utah.edu).}

\thanks{A.~Farhang and L.~E.~Doyle are with the CTVR/CONNECT, The Telecommunications Research Centre, Trinity College Dublin, Ireland, Dublin2 (e-mail: \{farhanga, ledoyle\}@tcd.ie).}  

\thanks{This publication has emanated from research supported in part by a research grant from Science Foundation Ireland (SFI) and is co-funded under the European Regional Development Fund under Grant Number 13/RC/2077.} }

\begin{document}

\maketitle

\begin{abstract}
We study the possibility of removing the cyclic prefix (CP) overhead from orthogonal frequency division multiplexing (OFDM) in massive multiple-input multiple-output (MIMO) systems. We consider the uplink transmission while our results are applicable to the downlink as well. The absence of CP increases the spectral efficiency in expense of intersymbol interference (ISI) and intercarrier interference (ICI). It is known that in massive MIMO, the effects of uncorrelated noise and multiuser interference vanish as the number of base station (BS) antennas tends to infinity.  To investigate if the channel distortions in the absence of CP fade away, we study the performance of the standard maximum ratio combining (MRC) receiver. Our analysis reveals that in this receiver, there always remains some residual interference leading to saturation of signal-to-interference-plus-noise ratio (SINR). To resolve this problem, we propose to use the time reversal (TR) technique. Moreover, in order to further reduce the multiuser interference, we propose a zero-forcing equalization to be deployed after the TR combining. We compare the achievable rate of the proposed system with that of the conventional CP-OFDM. We show that in realistic channels, a higher spectral efficiency is achieved by removing the CP from OFDM, while reducing the computational complexity.
\end{abstract}


\section{Introduction}\label{sec:Introduction}

\IEEEPARstart{M}{assive} multiple-input multiple-output (MIMO) is a multiuser technique enabling the users to simultaneously utilize the same resources in time and frequency. Massive MIMO significantly improves the capacity of the multiuser networks, making it a strong candidate technology for the fifth generation (5G) of cellular networks and a topic of interest for the research community, \cite{Marzetta2010, Larsson2014, Rusek2013}. In cases where the number of base station (BS) antennas is much larger than the number of users, optimal performance can be achieved through the most straightforward detection/precoding techniques, namely, maximum ratio combining/transmission, \cite{Marzetta2010}.

In the massive MIMO context, orthogonal frequency division multiplexing (OFDM) with cyclic prefix (CP) is particularly attractive because it enables the conversion of the frequency-selective channels between each mobile terminal (MT) antenna and the BS antennas into a set of flat-fading channels over each subcarrier band. Therefore, the MT data streams can be distinguished from each other through the respective channel responses. Hence, most of the literature deals with massive MIMO while utilizing OFDM with CP (CP-OFDM) \cite{Marzetta2010,Hoydis2013,ngo2013energy,Rusek2013}. However, the CP duration adds an extra overhead to the network and reduces the spectral efficiency. Therefore, in order to increase the transmission rate, it is desirable to eliminate the CP duration from OFDM. However, this comes at the expense of intersymbol interference (ISI) and intercarrier interference (ICI), imposed by the multipath channel. Here, it is worth mentioning that in massive MIMO, the effects of uncorrelated noise as well as various types of interference/imperfections such as multiuser interference (MUI), imperfect channel state information, hardware imperfections, phase noise, etc., will vanish as the number of BS antennas grows large \cite{Rusek2013,ngo2013energy,Bjornson2015,Pitarokoilis2015}. Therefore, the core question at the heart of this paper is:  

\emph{``Can massive MIMO average out the ISI and ICI introduced by the multipath channel in OFDM without CP?''}

There are a number of methods in the literature tackling the ISI and ICI problem of OFDM with insufficient CP, \cite{Molisch2007,toeltsch2000efficient,lim2006mimo,chen2009low,ma2009two}. References \cite{Molisch2007} and \cite{toeltsch2000efficient} suggest to remove the effect of ICI and ISI by utilizing the previously detected symbols and using successive interference cancellation (SIC). In \cite{lim2006mimo} and \cite{chen2009low}, a MIMO-OFDM scenario is considered, and iterative interference cancellation using turbo equalization is proposed. The authors in \cite{ma2009two} propose an interference cancellation algorithm based on some structural properties obtained from shifting the received OFDM blocks. We note that the above methods are designed for the conventional OFDM (or MIMO-OFDM) scenarios, and do not take advantage of the excessive number of BS antennas in a massive MIMO setup. In \cite{nsengiyumva2016cyclic}, the authors consider the conventional frequency-domain combining methods and deploy computer simulations to show that the CP duration can be shortened to achieve a higher spectral efficiency in a massive MIMO system. However, no detailed mathematical analysis of the proposed approach is presented.

In this paper, to investigate if the channel distortions (i.e., ISI and ICI) in the absence of CP fade away as the number of BS antennas grows large, we first study the performance of the conventional frequency-domain combining methods, such as maximum ratio combining (MRC), zero forcing (ZF), and minimum mean square error (MMSE) detectors, \cite{ngo2013energy}. We mathematically analyze the signal-to-interference-plus-noise ratio (SINR) performance of the above detectors when the CP is removed from the OFDM signal. Our SINR analysis reveals that when the above combining methods are applied, the channel distortions arising from the absence of CP, i.e., ISI and ICI, do not average out as the number of BS antennas tends to infinity. Thus, SINR saturates at a certain deterministic level and arbitrarily large SINR values cannot be achieved by increasing the BS array size.

To resolve the saturation issue, we propose to use a technique known as \emph{time reversal} (TR) to combine/precode the signals of different BS antennas in the \emph{time domain} instead of the \emph{frequency domain}. This technique is based on a pivotal phenomenon in physics that harnesses the principle of channel reciprocity and multipath effects to concentrate the signal energy at a certain point in space (\textit{spatial focusing}) and compress the channel impulse response in the time domain (\textit{temporal focusing}). This spatial-temporal focusing effect mitigates the ISI, ICI, and MUI, \cite{7497585}. Time reversal has been extensively studied and utilized in underwater acoustic channels, e.g., \cite{Edelmann2005,Gomes2008,gomes2004time,Rouseff2001}. In \cite{Zhiqiang2012} and \cite{liu2014overhead}, the authors utilize the temporal focusing property of TR and propose a CP length design method to satisfy specific performance requirements in underwater acoustic channels. The authors, consequently, balance the trade-off between the CP length and the resulting interference due to the residual ISI and ICI imposed by the insufficient length of CP. The scope of \cite{Zhiqiang2012} and \cite{liu2014overhead} is limited to small-scale underwater acoustic networks without any consideration of multiuser scenarios. Recently, there has been an emerging interest in the application of TR for the future generation of wireless networks, \cite{Chen2016}. Moreover, the application of TR to massive MIMO in the context of single-carrier transmission has been studied extensively, e.g., \cite{Pitarokoilis2012,Pitarokoilis2015,han2012time}. Application of TR to CP-OFDM has been studied in \cite{Dubois2013,Maaz2015,dubois2013performance}, where the authors consider a single-user massive MIMO scenario and show that TR can be applied to a CP-OFDM system either in the time or in the frequency domain. Moreover, the authors show that TR allows the CP length to be reduced thanks to its spatial-temporal focusing property.

As it is shown in \cite{Pitarokoilis2012} for the case of single-carrier transmission, with the TR technique, the channel distortions tend to zero as the number of BS antennas goes to infinity. We show that this result is also applicable to the case of OFDM without CP transmission. Thus, arbitrarily large SINR values can be achieved by increasing the BS array size. However, as we show in this paper, the performance of the conventional TR is limited due to the excessive amount of MUI when the number of BS antennas is finite. We show that OFDM allows for a straightforward zero-forcing equalization to be utilized after the TR combining. With this approach, the MUI level is significantly reduced and larger SINR values can be achieved compared to the conventional TR method, while the SINR saturation problem is also avoided. Throughout the paper, we refer to the conventional TR technique as \emph{TR-MRC}, while the proposed TR-based method with additional ZF equalization is referred to as \emph{TR-ZF}.

It is worth mentioning that in a typical communication system, a time period with duration equal to the coherence time of the channel is divided into two intervals: (i) training period, and (ii) data transmission period. In this paper, we only focus on the data transmission period and consider removing the CP overhead during this period. Throughout the paper, we consider perfect knowledge of the channel state information (CSI) at the BS and assume that the CP is included in the course of training to establish the carrier frequency and timing synchronization and obtain an accurate CSI. Studying the problem of CP removal/shortening during the training interval in the context of massive MIMO remains for the future study.

Also, in this paper, we focus on the uplink transmission, but the results and algorithms are trivially applicable to the downlink as well. We analytically derive the SINR performance of the TR-MRC receiver as well as our proposed TR-ZF technique. Based on our SINR derivations, we obtain a lower bound on the achievable information rate for both the TR-MRC and TR-ZF receivers. We show that higher spectral efficiency can be achieved using OFDM without CP as compared to CP-OFDM. More specifically, we show that using TR-MRC and TR-ZF techniques in OFDM without CP, higher information rate is achievable as compared to the case of CP-OFDM with the conventional MRC and ZF detection methods, respectively. Furthermore, we analyze the computational complexity of both TR-MRC and TR-ZF methods and introduce computationally efficient ways to implement them. We show that while the complexity of TR-MRC is almost similar to the frequency-domain MRC, a significantly lower complexity is obtained when utilizing the TR-ZF equalizer as compared to the conventional ZF detector.

To summarize, we list the contributions of this paper as follows:
\begin{itemize}
\item To increase the spectral efficiency of massive MIMO systems, we show that the CP overhead can be successfully eliminated. This is a result of the coherent combining of the received signals at the BS antennas that yields the channel distortions to disappear as the BS array size increases.
\item We show that in the absence of the CP, the SINR performance of the conventional frequency-domain combining methods, i.e., MRC, ZF, and MMSE, saturates at a certain deterministic level. Hence, arbitrarily large SINR values cannot be achieved by increasing the array size at the BS.
\item We propose to use the TR technique to resolve the above saturation problem. 
\item Although the conventional TR technique can achieve reasonable SINR values and is a viable option in many scenarios, it suffers from a high level of MUI in multiuser cases. We propose a novel ZF equalization technique to be applied after the TR operation to reduce the MUI level.
\item We introduce efficient methods to minimize the computational cost of both TR-MRC and TR-ZF receivers. We also compare the complexity of the proposed receiver structures with that of the conventional CP-OFDM with MRC and ZF detectors.
\item We perform a thorough analysis and obtain closed-form expressions for the SINR and achievable rate performance of both TR-MRC and TR-ZF receivers.
\end{itemize}

The rest of the paper is organized as follows. After presenting the system model in Section~\ref{sec:SystemModel}, we discuss the saturation problem of the conventional frequency-domain combiners that arises in the absence of CP in Section \ref{sec:ConventionalEq}. The TR technique is introduced as a remedy to this problem in Section \ref{sec:TR}, where we also propose a novel ZF post-equalization to further reduce the MUI. In Section \ref{sec:implementation}, we present a complexity analysis of the receiver structures that are introduced in this paper, and compare them with that of the conventional CP-OFDM. The asymptotic performance, in terms of SINR and achievable rate, of the TR-MRC and the proposed TR-ZF receivers is analyzed in Section \ref{sec:tr_sinr}. Our discussions in this paper are numerically evaluated in Section \ref{sec:numerical_results}. Finally, we conclude the paper in Section \ref{sec:conclusion}.

\textit{Notations:} Matrices, vectors and scalar quantities are denoted by boldface uppercase, boldface lowercase, and normal letters, respectively. $[\bA]_{mn}$ represents the element in the $m^{\rm{th}}$ row and $n^{\rm{th}}$ column of $\bA$ and $\bA^{-1}$ signifies the inverse of $\bA$. $\I_M$ is the identity matrix of size $M\times M$, and $\bzero_{M\times N}$ is the zero matrix of size $M \times N$. The matrix trace operation is denoted by $\trace \{\cdot \}$. $\bD={\rm diag}\{\ba\}$ represents a diagonal matrix whose diagonal elements are formed by the elements of the vector $\ba$. The superscripts $(\cdot)^{\rm T}$, $(\cdot)^{\rm H}$ and $(\cdot)^\ast$ indicate transpose, conjugate transpose, and conjugate operations, respectively.  The linear convolution is denoted by $\star$. $\mathbb{E}\{\cdot\}$ denotes the expected value of a random variable. The notation $\mathcal{CN}(0,\sigma^2)$ represents the circularly-symmetric and zero-mean complex normal distribution with the variance of $\sigma^2$. Throughout the paper, frequency-domain variables are signified by over-bar accent. 

\section{System Model}\label{sec:SystemModel}

We consider a large-scale multiuser MIMO system similar to the one discussed in \cite{Marzetta2010}. For the sake of simplicity, only a single-cell scenario is considered. Also, in this paper, only the case of uplink transmission is discussed but the results and algorithms are trivially applicable to the case of downlink transmission as well. We consider $K$ mobile terminals that are simultaneously communicating with a BS which is equipped with an array of $M$ antenna elements. Each MT is a single-antenna device. In this paper, we consider an asymptotic regime where the number of BS antennas $M$ tends to infinity.

We consider a discrete-time model for our analysis. A similar model is also considered in \cite{Molisch2007,toeltsch2000efficient,lim2006mimo,chen2009low} to study OFDM without CP or with insufficient CP. Let $x_k(l)$ represent the transmit signal of $k^{\rm th}$ terminal in discrete time. Thus, the received signal at the $m^{\rm th}$ BS antenna can be obtained as
\be\label{eqn:rm} 
r_m(\ell) = \sum^{K-1}_{k=0} x_k(\ell) \star h_{m,k}(\ell) + \nu_m(\ell),
\ee
where $\nu_m(\ell)$ is the complex additive white Gaussian noise (AWGN) at the input of $m^{\rm th}$ BS antenna, and is distributed according to $\nu_m(\ell)~{\sim}~{\mathcal {CN}}(0,\sigma_\nu^{2})$, where ${\sigma_\nu}^{2}$ is the noise variance. The sequence $h_{m,k}(\ell)$ represents the channel impulse response (CIR) between terminal $k$ and BS antenna $m$. Throughout this paper, we assume that the BS has a perfect knowledge of the CSI. The channel impulse responses are modeled as time-invariant filters with the length of $L$, and independent channel responses are assumed between each MT antenna and the BS antennas. The multipath channel tap $h_{m,k}(\ell)$, for $\ell \in \{0,1,\dots,L-1\}$, follows the ${\mathcal {CN}}(0,\rho(\ell))$ distribution, and different taps are assumed to be independent with each other. Here, $\rho(\ell)$ represents the power delay profile (PDP) of the channel model. Throughout this paper, normalized channel PDP is considered, i.e. $\sum_{\ell=0}^{L-1}\rho(\ell)=1$. We also assume that the average power of the signal transmitted by each MT is equal to one, i.e., $\mathbb{E}\{|x_k(\ell)|^2\}=1$.  Accordingly, $1/\sigma_\nu^2$ is the average signal-to-noise ratio (SNR) at the BS input.

In this paper, we assume OFDM modulation is used for data transmission with the total number of $N$ subcarriers. To increase the bandwidth efficiency, we do not insert CP/guard interval between the successive OFDM symbols. Therefore, the $i^{\rm th}$ OFDM symbol of terminal $k$ can be obtained as $\x_{k,i}=\bF_N^{\rm H}\bd_{k,i}$, where $\bF_N$ is the $N$-point normalized discrete Fourier transform (DFT) matrix, and $\bd_{k,i}=[d_{k,i}(0),\ldots,d_{k,i}(N-1)]^{\rm T}$ is the transmit data vector of terminal $k$ on symbol time index $i$. The elements of $\bd_{k,i}$ are independent and identically distributed (i.i.d.) zero-mean complex random variables with the variance of unity. Assuming that the number of transmitted OFDM symbols is $Q$, the vector of transmit signal of $k^{\rm th}$ terminal, $\x_k$, is obtained by concatenation of different OFDM symbols, i.e., $\x_k = [\x_{k,0}^{\rm T},\ldots,\x_{k,Q-1}^{\rm T}]^{\rm T}$.

\section{Frequency-Domain Combining Approach}\label{sec:ConventionalEq}

Conventionally, in CP-OFDM systems, MRC, ZF and MMSE combiners are applied in the frequency domain. With such a setup and in a large-scale multiuser MIMO scenario, the multiuser interference and noise effects average out as the number of BS antennas tends to infinity, \cite{Marzetta2010}. Hence, SINR increases without any bound as the number of BS antennas increases. In the case of interest to this paper, i.e., in the absence of CP, SINR saturation occurs, and thus, arbitrary large information rates cannot be achieved by increasing the BS antennas. In this section, we dig into the mathematical details that explain this limitation of the conventional frequency-domain combiners when applied to the OFDM without CP signal. In the next section, we introduce the TR combining as a remedy to this problem.

Let us consider the equalization of the $i^{\rm th}$ OFDM symbol. To this end, we form the $N \times 1$ vector $\br_{m,i} = [r_m(iN),\dots, r_m(iN+N-1)]^{\rm T}$ by considering the $i^{\rm th}$ segment of the signal $r_m(\ell)$, and follow \cite{Molisch2007,toeltsch2000efficient,lim2006mimo,chen2009low} to express (\ref{eqn:rm}) in the matrix form as 
\be\label{eqn:rmi}
\br_{m,i} = \sum^{K-1}_{k=0} \bH_{m,k}^{(i,i-1)}\x_{k,i-1}+\bH_{m,k}^{(i,i)}\x_{k,i}+{\bnu}_{m,i},
\ee
where,
\bse
\begin{align} \setlength{\arraycolsep}{2.5pt}
\footnotesize
 \bH_{m,k}^{(i,i-1)} \hspace{-3pt} = \hspace{-3pt}
 \begin{pmatrix}
  0 & \cdots & h_{m,k}(L-1) & h_{m,k}(L-2) & \cdots & h_{m,k}(1) \\
	0 & \cdots & 0   & h_{m,k}(L-1) & \cdots & h_{m,k}(2)  \\
  \vdots & \ddots  & \vdots & \vdots & \ddots & \vdots \\
  0      & \cdots  &  \cdots     & \cdots & \cdots & h_{m,k}(L-1)\\
	0      & \cdots & \cdots  &  \cdots     & \cdots & 0\\
	\vdots  & \vdots  & \vdots & \vdots & \ddots & \vdots \\
	0 & \cdots & \cdots  &  \cdots & \cdots & 0 
 \end{pmatrix}, \label{eqn:H0_ISI} 
\end{align}
\begin{align} \setlength{\arraycolsep}{2.5pt}
\footnotesize
 \bH_{m,k}^{(i,i)} \hspace{-2pt} = \hspace{-2pt}
 \begin{pmatrix}
  h_{m,k}(0) & 0   & 0 & \cdots & 0 \\
	h_{m,k}(1) & h_{m,k}(0) & 0 & \cdots & 0 \\
  \vdots  & \vdots  & \vdots & \ddots & \vdots \\
  h_{m,k}(L-1) & h_{m,k}(L-2) & h_{m,k}(L-3) & \cdots & 0 \\
	0 & h_{m,k}(L-1) & h_{m,k}(L-2) & \cdots & 0 \\
	\vdots  & \vdots  & \vdots & \ddots & \vdots \\
	0 & 0 & 0 & \cdots & h_{m,k}(0)  
 \end{pmatrix}. \label{eqn:H0_ICI}
\end{align}
\ese
The $N \times N$ convolution matrices $\bH_{m,k}^{(i,i-1)}$ and $\bH_{m,k}^{(i,i)}$, when multiplied to the vectors $\x_{k,i-1}$ and $\x_{k,i}$, create the tail of the symbol $i-1$ overlapping with $L-1$ samples in the beginning of the symbol $i$ and the channel affected symbol $i$, respectively. The vector ${\bnu}_{m,i}$ includes $N$ samples of the AWGN signal $\nu_m(\ell)$ at the position of symbol $i$.

Next, the received signals at different BS antennas are passed through OFDM demodulators (DFT blocks), and then, the outputs of the DFT blocks across different BS antennas are combined using the frequency-domain channel coefficients between the terminals and BS antennas. To cast this procedure into a mathematical formulation and pave the way for our analysis, we obtain the output of the OFDM demodulator at BS antenna $m$ as
\begin{align}\label{eqn:rmibar}
\bar{\br}_{m,i} \nonumber &= \hspace{-3pt} \sum^{K-1}_{k=0} \hspace{-2pt} \Big(\bF_N\bH_{m,k}^{(i,i-1)} \x_{k,i-1}+\bF_N\bH_{m,k}^{(i,i)}\x_{k,i}\Big)+\bF_N{\bnu}_{m,i} \nonumber \\ 
&= \hspace{-3pt} \sum^{K-1}_{k=0} \hspace{-2pt} \Big(\bF_N\bH_{m,k}^{(i,i-1)} \bF_N^{\rm H}\bd_{k,i-1} \hspace{-2pt} + \hspace{-2pt} \bF_N\bH_{m,k}^{(i,i)} \bF_N^{\rm H}\bd_{k,i}\Big) \hspace{-3pt} + \hspace{-2pt} \bar{\bnu}_{m,i} \nonumber\\
&= \hspace{-3pt} \sum^{K-1}_{k=0} \hspace{-2pt} \Big(\xbar{\bH}_{m,k}^{(i,i-1)} \bd_{k,i-1} + \xbar{\bH}_{m,k}^{(i,i)} \bd_{k,i}\Big)+\bar{\bnu}_{m,i},
\end{align}
where the bar symbol in $\bar{\br}_{m,i}$ and $\bar{\bnu}_{m,i}$ is to indicate that they are in the frequency domain. Similarly, the matrices $\xbar{\bH}_{m,k}^{(i,i-1)} \triangleq \bF_N\bH_{m,k}^{(i,i-1)} \bF_N^{\rm H}$ and $\xbar{\bH}_{m,k}^{(i,i)} \triangleq \bF_N\bH_{m,k}^{(i,i)} \bF_N^{\rm H}$ are the frequency-domain intersymbol and intercarrier interference matrices, respectively. Note that in the case of CP-OFDM transmission, $\xbar{\bH}_{m,k}^{(i,i-1)}=\bzero_{N \times N}$ and $\xbar{\bH}_{m,k}^{(i,i)}$ is a diagonal matrix with the diagonal entries given by the frequency-domain channel coefficients, i.e., $\big[\xbar{\bH}_{m,k}^{(i,i)}\big]_{pp} = \bar{h}_{m,k}(p) \triangleq \sum_{\ell = 0}^{L-1} h_{m,k}(\ell) e^{-j\frac{2\pi \ell p}{N}}$.

Let $\BW_p$ be the $M\times K$ combining matrix corresponding to subcarrier $p\in\{0,\ldots,N-1\}$, and the $M \times 1$ vector $\bar{\br}_i(p) = [\bar{r}_{0,i}(p),\dots,\bar{r}_{M-1,i}(p)]^{\rm T}$ contain the $p^{\rm th}$ outputs of the DFT blocks at different BS antennas. Accordingly, the output of the combiner can be obtained as
\be \label{eqn:dhati}
\hat{\bd}_i(p) = \BW_p^{\rm H}  \hspace{2pt} \bar{\br}_i(p) ,
\ee
where the $K \times 1$ vector $\hat{\bd}_i(p) = [\hat{d}_{0,i}(p),\ldots,\hat{d}_{K-1,i}(p)]^{\rm T}$ contains the detected symbols of all terminals at subcarrier $p$ and time index $i$. We consider three conventional linear combiners, namely, MRC, ZF and MMSE. For these combiners, we have
\be \label{eqn:mrc_zf_mmse}
\BW_p = 
\begin{cases}
	\xbar{\bH}_p \bD_p^{-1} ,  & {\rm for~ MRC}, \\
	\xbar{\bH}_p \left( \xbar{\bH}_p^{\rm H} \xbar{\bH}_p \right)^{-1},	   & {\rm for~ ZF}, \\
	\xbar{\bH}_p \left( \xbar{\bH}_p^{\rm H} \xbar{\bH}_p + \sigma_\nu^2 \eye_K \right)^{-1},	   & {\rm for~ MMSE}, 
\end{cases}
\ee
where $\xbar{\bH}_p$ is the $M\times K$ matrix of frequency-domain channel coefficients for the $p^{\rm th}$ subcarrier, i.e., $\left[\xbar{\bH}_p\right]_{mk} = \bar{h}_{m,k}(p) \triangleq \sum_{\ell = 0}^{L-1} h_{m,k}(\ell) e^{-j\frac{2\pi \ell p}{N}}$. In the case of MRC, $\bD_p$ is a $K\times K$ diagonal matrix whose diagonal elements are formed by the diagonal elements of $\xbar{\bH}_p^{\rm H} \xbar{\bH}_p $. The role of $\bD_p$ is just to normalize the amplitude of the MRC output. Without this term, the amplitude grows linearly without a bound as the number of BS antennas increases.

We note that for large number of BS antennas $M$ and using the law of large numbers, $\frac{1}{M} \xbar{\bH}_p^{\rm H} \xbar{\bH}_p$ tends to $ \eye_K$, \cite{ngo2013energy}. Similarly, the matrix $\frac{1}{M} \bD_p$ tends to $\eye_K$ as the number of BS antennas increases. In light of this observation, in the following, to find the various interference terms in the large-antenna regime, we consider $\BW_p = \frac{1}{M} \xbar{\bH}_p$.

Following (\ref{eqn:rmibar}) and (\ref{eqn:dhati}), the detected symbol $\hat{d}_{k,i}(p)$ can be expressed as
\begin{align} \label{eqn:dhat_expanded}
&\hat{d}_{k,i}(p) = \nonumber \\
& \underbrace{ \mathcal{H}_{kk,pp}^{(i,i)}  d_{k,i}(p) \vphantom{\sum_{\substack{q = 0 \\ q \neq p}}}}_{\text{Desired Signal}} + \underbrace{\sum_{\substack{q = 0\\q \neq p}}^{N-1} \mathcal{H}_{kk,pq}^{(i,i)}  d_{k,i}(q)}_{\text{ICI}} + \underbrace{\sum_{q=0}^{N-1} \mathcal{H}_{kk,pq}^{(i,i-1)}  d_{k,i-1}(q) \vphantom{\sum_{\substack{q = 0 \\ q \neq p}}}}_{\text{ISI}} \nonumber \\
&+ \underbrace{\sum_{\substack{j = 0 \\ j \neq k}}^{K-1} \sum_{q=0}^{N-1} \Big( \mathcal{H}_{kj,pq}^{(i,i-1)}  d_{j,i-1}(q) + \mathcal{H}_{kj,pq}^{(i,i)}  d_{j,i}(q) \Big)}_{\text{MUI}} + \underbrace{\bar{\nu}_{k,i}(p) \vphantom{\sum_{\substack{j = 0 \\ j \neq k}}}}_{\text{Noise}} ,
\end{align}
where the interference coefficients $\mathcal{H}_{kj,pq}^{(i,i-1)}$ and $\mathcal{H}_{kj,pq}^{(i,i)}$ determine the amount of the interference from symbols $d_{j,i-1}(q)$ and $d_{j,i}(q)$, respectively, on the detected symbol $\hat{d}_{k,i}(p)$. These interference coefficients capture the effects of the combiner gains together with the ICI and ISI coefficients in (\ref{eqn:rmibar}). Mathematically, we can calculate $\mathcal{H}_{kj,pq}^{(i,i)}$ and $\mathcal{H}_{kj,pq}^{(i,i-1)}$ according to
\bse \label{eqn:lambdas_def}
\begin{align} 
\mathcal{H}_{kj,pq}^{(i,i)}  &= \hspace{-0mm} \frac{\bar{\bh}_{k}^{\rm H}(p)}{M}  \left[\left[\xbar{\bH}_{0,j}^{(i,i)}\right]_{pq}\hspace{-1.0mm},\ldots,\left[\xbar{\bH}_{M-1,j}^{(i,i)}\right]_{pq}\right]^{\rm T}, \\
\mathcal{H}_{kj,pq}^{(i,i-1)} &= \hspace{-0mm} \frac{\bar{\bh}_{k}^{\rm H}(p)}{M}  \left[\left[\xbar{\bH}_{0,j}^{(i,i-1)}\right]_{pq}\hspace{-1.0mm},\ldots,\left[\xbar{\bH}_{M-1,j}^{(i,i-1)}\right]_{pq}\right]^{\rm T} ,
\end{align}
\ese
respectively. The $M \times 1$ vector $\bar{\bh}_{k}(p) = [\bar{h}_{0,k}(p),\dots,\bar{h}_{M-1,k}(p)]$ is the $k^{\rm th}$ column of the matrix $\xbar{\bH}_p$ containing the frequency-domain channel coefficients between terminal $k$ and different BS antennas.

Before we proceed, we review some results from probability theory. Let $\ba = [a_1,\dots,a_n]^{\rm T}$ and $\bb = [b_1,\dots,b_n]^{\rm T}$ be two random vectors each  containing i.i.d. elements. Furthermore, assume that $i^{\rm th}$ elements of $\ba$ and $\bb$ are correlated according to $\mathbb{E}\big\{ a_i^* b_i \big\} = C_{ab}$, $i = 1,\dots,n$.  Then, according to the law of large numbers, the sample mean $\frac{1}{n} \ba^{\rm H} \bb = \frac{1}{n}\sum_{i=1}^n a_i b_i$ converges almost surely to the distribution mean $C_{ab}$ as $n$ tends to infinity, i.e.,
\be 
\frac{1}{n} \ba^{\rm H} \bb  \rightarrow C_{ab}  \hspace{5pt} \text{ as } \hspace{5pt} n \rightarrow \infty , 
\ee
with almost sure convergence.

Using the law of large numbers, the interference coefficients given in (\ref{eqn:lambdas_def}) converge almost surely to the following values as the number of BS antennas $M$ tends to infinity.
\bse \label{eqn:lambda_asym}
\bea
\mathcal{H}_{kj,pq}^{(i,i)}  &\rightarrow&  \mathbb{E} \Big\{ \bar{h}_{m,k}^\ast(p) \left[\xbar{\bH}_{m,j}^{(i,i)}\right]_{pq} \Big\}, \label{eqn:lambda_asym_a}\\
\mathcal{H}_{kj,pq}^{(i,i-1)} &\rightarrow&  \mathbb{E} \Big\{ \bar{h}_{m,k}^\ast(p) \left[\xbar{\bH}_{m,j}^{(i,i-1)}\right]_{pq} \Big\}. \label{eqn:lambda_asym_b}
\eea
\ese
Note that the asymptotic values in (\ref{eqn:lambda_asym_a}) and (\ref{eqn:lambda_asym_b}) are the statistical correlation of the combiner tap value $\bar{h}_{m,k}(p)$ with the interference components $\left[\xbar{\bH}_{m,j}^{(i,i)}\right]_{pq}$ and $\left[\xbar{\bH}_{m,j}^{(i,i-1)}\right]_{pq}$, respectively.

In the Appendix, we have simplified the expressions in (\ref{eqn:lambda_asym}). The result is that as $M$ grows large, the coefficients $\mathcal{H}_{kj,pq}^{(i,i)}$ and $\mathcal{H}_{kj,pq}^{(i,i-1)}$ for $k \neq j$ tend to zero. Accordingly, the MUI term in (\ref{eqn:dhat_expanded}) fades away asymptotically. On the other hand, the ICI and ISI terms remain as specified according to the following coefficients:
\bse\label{eqn:lambdas_simplified} 
\begin{alignat}{2}
&\mathcal{H}_{kk,pp}^{(i,i)}   && \rightarrow~ 1-\frac{{\tau_{\rm av}}}{N}, \label{eqn:lambdas_simplified_a}\\
&\mathcal{H}_{kk,pq}^{(i,i)}   && \rightarrow~ \frac{1-\bar{\rho}(q-p)}{N(1-e^{j\frac{2\pi(q-p)}{N}})}, ~~~ ({\rm for}~ p\neq q), \label{eqn:lambdas_simplified_b} \\
&\mathcal{H}_{kk,pp}^{(i,i-1)} && \rightarrow~ \frac{{\tau_{\rm av}}}{N}, \label{eqn:lambdas_simplified_c} \\
&\mathcal{H}_{kk,pq}^{(i,i-1)} && \rightarrow~ \frac{\bar{\rho}(q-p)-1}{N(1-e^{j\frac{2\pi (q-p)}{N}})}, ~~~ ({\rm for}~ p\neq q), \label{eqn:lambdas_simplified_d}
\end{alignat}
\ese
where ${\tau_{\rm av}} \triangleq \sum^{L-1}_{\ell=0}\ell \rho(\ell)$, is the average delay spread of the channel, and $\bar{\rho}(q) \triangleq \sum_{\ell=0}^{N-1} \rho(\ell) e^{-j\frac{2\pi \ell q}{N}}$.

\begin{proposition} \label{prp:saturation1}
In the absence of CP and with the conventional MRC, ZF, or MMSE combiners, as the number of BS antennas tends to infinity, SINR for each terminal converges almost surely to
\bea\label{eqn:SIR_sat}
{\rm SINR} ~ \rightarrow ~ \frac{\big(1-\frac{\tau_{\rm av}}{N}\big)^2}{\big(\frac{\tau_{\rm av}}{N}\big)^2+\sum\limits^{N-1}_{\eta=1}\frac{|1-\bar{\rho}(\eta)|^2}{2N^2\sin^2(\pi \eta/N)}}.
\eea
Hence, SINR saturation occurs and arbitrary large SINR values cannot be achieved by increasing the number of BS antennas.
\end{proposition}
\begin{proof}
As the number of BS antennas $M$ tends to infinity, the coefficients $\mathcal{H}_{kj,pq}^{(i,i)}$ and $\mathcal{H}_{kj,pq}^{(i,i-1)}$ for $k \neq j$ tend to zero; see the Appendix. Hence, the contribution of multiuser interference becomes negligible. A similar argument can be developed for the noise contribution. Thus, the SINR of terminal $k$ at subcarrier $p$ is determined based on the ICI and ISI terms and can be calculated as 
\begin{equation} 
{\rm SINR}_{k,p} = \frac{\Big| \mathcal{H}_{kk,pp}^{(i,i)} \Big|^2}{ \sum_{\substack{ q=0 \\ q\neq p}}^{N-1} \Big| \mathcal{H}_{kk,pq}^{(i,i)} \Big|^2 + \sum_{q=0}^{N-1} \Big| \mathcal{H}_{kk,pq}^{(i,i-1)} \Big|^2}  .
\end{equation}
This reduces to (\ref{eqn:SIR_sat}), following (\ref{eqn:lambdas_simplified_a}) through (\ref{eqn:lambdas_simplified_d}) and noting that the asymptotic SINR value is equal for all terminals and all subcarriers.
\end{proof}

We note that although the analysis in this section was based on the OFDM without CP, one can follow a similar line of derivations to show that, in general, when insufficient CP lengths are utilized, the SINR saturation problem occurs.
\section{Time-Reversal and Equalization} \label{sec:TR}

As it was shown in the previous section, when CP is removed from the OFDM signal, the conventional frequency-domain combining methods lead to some residual ICI and ISI components that will not fade away even with infinite number of BS antennas. Consequently, SINR saturates at a certain deterministic level. In order to resolve this problem, in this section, we propose to use TR to combine the signals of different BS antennas in the \emph{time domain} instead of the \emph{frequency domain}. As it is shown in \cite{Pitarokoilis2012} for the case of single-carrier transmission, with TR combining, intersymbol interference and multiuser interference tend to zero as the number of BS antennas goes to infinity. Thus, arbitrarily large SINR values can be achieved by increasing the BS array size. However, as we show in this paper, performance of the conventional TR combining is rather limited due to the excessive amount of multiuser interference when the number of BS antennas is finite. We show that OFDM allows for a straightforward zero-forcing equalization to be utilized after the TR combining. With this approach, the MUI level is significantly reduced  and larger SINR values can be achieved compared to the conventional TR method, while the saturation problem is also resolved. A more detailed discussion on the TR-MRC and TR-ZF receivers is presented in the following subsections.

\subsection{TR-MRC} \label{sec:TRMRC}

In TR-MRC, for a given terminal, e.g. $k^{\rm th}$ terminal, the received signals at the BS antennas are first prefiltered with the time-reversed and conjugated versions of the CIRs between that terminal and the corresponding BS antennas. Then, the resulting signals are combined with each other. Using (\ref{eqn:rm}), this procedure can be mathematically written as
\begin{align}
r^{\tr}_{k}(\ell) &= \frac{1}{\sqrt{M}}\sum^{M-1}_{m=0} r_m(\ell) \star h_{m,k}^*(-\ell) \nonumber \\
&= \sum^{K-1}_{j=0} x_j(\ell) \star g_{kj}(\ell) + \nu^{\tr}_{k}(\ell), \label{eqn:rm2} 
\end{align}
where
\be \label{eqn:gkj}
g_{kj}(\ell) \triangleq \frac{1}{\sqrt{M}} \sum^{M-1}_{m=0} h_{m,j}(\ell) \star h^*_{m,k}(-\ell) ,
\ee
is the equivalent CIR after the TR operation. In particular, $g_{kj}(\ell)$, for $j \neq k$, is the cross-talk CIR between the terminals $k$ and $j$, and $g_{kk}(\ell)$ is the time-reversal equivalent CIR of terminal $k$. Also, $\nu^{\tr}_{k}(\ell) \triangleq \frac{1}{\sqrt{M}} \sum_{m=0}^{M-1} \nu_m(\ell) \star h^{*}_{m,k}(-\ell)$ is the noise contribution after the TR operation.

Here, we focus on the equalization of the $i^{\rm th}$ OFDM symbol. Hence, let the $N \times 1$ vector $\br^{\tr}_{k,i} = [r^{\tr}_{k}(iN),\dots,r^{\tr}_{k}(iN+N-1)]^{\rm T}$ contain the $i^{\rm th}$ segment of the signal $r^{\tr}_{k}(\ell)$. Accordingly, (\ref{eqn:rm2}) can be expressed in a matrix form as
\be\label{eqn:rmi1}
\br^{\tr}_{k,i} = \hspace{-0.15cm} \sum^{K-1}_{j=0}\hspace{-0.0cm} \big(\bG_{kj}^{(i,i-1)}\x_{j,i-1}+\bG_{kj}^{(i,i)}\x_{j,i}+\bG_{kj}^{(i,i+1)}\x_{j,i+1}\big)+\bnu^{\tr}_{k,i},
\ee
where the vector ${\bnu}_{k,i}^\tr$ includes $N$ samples of the AWGN signal $\nu_k^\tr(\ell)$ at the position of symbol $i$. The matrices $\bG_{kj}^{(i,i-1)}$, $\bG_{kj}^{(i,i)}$ and $\bG_{kj}^{(i,i+1)}$ are $N\times N$ convolution matrices comprising the ISI components due to the tail of the symbol $i-1$, the ICI components within the symbol $i$ and the ISI components originating from the beginning of the symbol $i+1$, respectively. The matrices $\bG_{kj}^{(i,i-1)}$ and $\bG_{kj}^{(i,i)}$ can be defined in a similar way as in (\ref{eqn:H0_ISI}) and (\ref{eqn:H0_ICI}), respectively. Here, we use the following compact notation:
\begin{alignat}{2}
&\bG_{kj}^{(i,i-1)} &&= \Toeplitz_{N \times N}\big(\big[g_{kj}(1),\ldots ,g_{kj}(L-1),\bzero_{1\times 2N-L}\big]^\mathrm{T}\big), \nonumber \\
&\bG_{kj}^{(i,i)}   &&= \Toeplitz_{N \times N}\big(\big[\bzero_{1\times N-L},\bg_{kj},\bzero_{1\times N-L}\big]^\mathrm{T}\big), \nonumber \\
&\bG_{kj}^{(i,i+1)} &&= \Toeplitz_{N \times N}\big(\big[\bzero_{1\times 2N-L},g_{kj}(1-L),\ldots ,g_{kj}(-1)\big]^\mathrm{T}\big), \label{eqn:Gtime}
\end{alignat}
where, $\bg_{kj} \triangleq \big[ g_{kj}(1-L),\dots,g_{kj}(L-1) \big]^{\rm T}$ contains the samples of the TR channel impulse response $g_{kj}(\ell)$. The notation $\bA = \Toeplitz_{M \times N} (\ba)$ for an ${(N+M-1) \times 1}$ vector $\ba$, represents an $M\times N$ Toeplitz matrix, in which $[\bA]_{mn} = [\ba]_{m-n+N}$. Accordingly, the vector $\ba$ is formed by starting from the top right element of $\bA$, going along the first row to the top left element and then going along the first column to the bottom left element.

Applying an $N$-point DFT block to $\br^\tr_{k,i}$, we obtain the following frequency-domain signal.
\begin{align}\label{eqn:rbarTR}
\bar{\br}_{k,i}^{\tr} \hspace{-2pt}= \hspace{-4pt}
\sum_{j=0}^{K-1} \hspace{-3pt} \Big( \xbar{\bG}_{kj}^{(i,i-1)} \bd_{j,i-1} \hspace{-2pt} + \hspace{-2pt} \xbar{\bG}_{kj}^{(i,i)} \bd_{j,i} \hspace{-2pt}+\hspace{-2pt} \xbar{\bG}_{kj}^{(i,i+1)} \bd_{j,i+1} \Big) \hspace{-4pt} + \hspace{-2pt}\bar{\bnu}_{k,i}^{\tr},
\end{align} 
where $\xbar{\bG}_{kj}^{(i,i-1)} \triangleq \bF_N\bG_{kj}^{(i,i-1)}\bF_N^{\rm H}$, $\xbar{\bG}_{kj}^{(i,i)} \triangleq \bF_N\bG_{kj}^{(i,i)}\bF_N^{\rm H}$, $\xbar{\bG}_{kj}^{(i,i+1)} \triangleq \bF_N\bG_{kj}^{(i,i+1)}\bF_N^{\rm H}$ and $\bar{\bnu}_{k,i}^{\tr} \triangleq \bF_N\bnu^{\tr}_{k,i}$. Let $\bar{r}_{k,i}^{\tt TR}(p)$ and $\bar{\nu}_{k,i}^{\tt TR}(p)$ be the $p^{\rm th}$ elements of the vectors $\bar{\br}_{k,i}^{\tr}$ and $\bar{\bnu}_{k,i}^{\tr}$, respectively. Moreover, we define $\mathcal{G}_{kj,pq}^{(i,i-1)}$, $\mathcal{G}_{kj,pq}^{(i,i)}$, and $\mathcal{G}_{kj,pq}^{(i,i+1)}$ as the elements $pq$ of the matrices $\xbar{\bG}_{kj}^{(i,i-1)}$, $\xbar{\bG}_{kj}^{(i,i)}$, and $\xbar{\bG}_{kj}^{(i,i+1)}$, respectively. Hence, $\bar{r}_{k,i}^{\tt TR}(p)$ can be expressed as in (\ref{eqn:rbarTR_expanded}) on the top of the next page. In Section \ref{sec:tr_sinr}, we analyze the interference terms given in (\ref{eqn:rbarTR_expanded}) and show that in this case, the SINR will grow without a bound as $M$ grows large. Consequently, the SINR saturation problem is resolved through deployment of TR-MRC.

\begin{figure*}[!t]
\begin{align} \label{eqn:rbarTR_expanded}
\bar{r}_{k,i}^{\tt TR}(p) &= 
\underbrace{\mathcal{G}_{kk,pp}^{(i,i)} \hspace{1pt} d_{k,i}(p) \vphantom{\sum_{\substack{q = 0 \\ q \neq p}}}}_{\text{Desired Signal}} + \underbrace{\sum_{\substack{q = 0\\q \neq p}}^{N-1} \mathcal{G}_{kk,pq}^{(i,i)} \hspace{1pt} d_{k,i}(q)}_{\text{ICI}} + \underbrace{\sum_{q=0}^{N-1} \Big( \mathcal{G}_{kk,pq}^{(i,i-1)} \hspace{1pt} d_{k,i-1}(q) + \mathcal{G}_{kk,pq}^{(i,i+1)} \hspace{1pt} d_{k,i+1}(q) \Big) \vphantom{\sum_{\substack{q = 0 \\ q \neq p}}}}_{\text{ISI}} \nonumber \\
&+ \underbrace{\sum_{\substack{j = 0 \\ j \neq k}}^{K-1} \sum_{q=0}^{N-1} \Big( \mathcal{G}_{kj,pq}^{(i,i-1)} \hspace{1pt} d_{j,i-1}(q) + \mathcal{G}_{kj,pq}^{(i,i)} \hspace{1pt} d_{j,i}(q) + \mathcal{G}_{kj,pq}^{(i,i+1)} \hspace{1pt} d_{j,i+1}(q) \Big)}_{\text{MUI}} + \underbrace{\bar{\nu}_{k,i}^{\tt TR}(p) \vphantom{\sum_{\substack{j = 0 \\ j \neq k}}}}_{\text{Noise}} .
\end{align}
\hrulefill
\end{figure*}

\subsection{ZF Post-Equalization (TR-ZF)} \label{sec:TRZF}

As mentioned in Section \ref{sec:TRMRC}, the SINR saturation problem is resolved through deployment of TR-MRC. Hence, as the number of BS antennas grows large, the power of different interference terms tends to zero and arbitrarily large SINR values can be achieved. However, for \emph{finite} number of BS antennas, this receiver suffers from a significant amount of interference in multiuser networks. This is mainly due to the interference originating from the symbols of different terminals transmitted on the same time and frequency slots. To gain a better intuition, we note that the TR-MRC receiver can be analogous to the MRC receiver used in CP-OFDM systems. The MRC receiver is simple and allows for arbitrarily large SINR values in CP-OFDM systems by increasing the number of BS antennas. However, multiuser interference is an important issue in MRC. Therefore, to tackle the multiuser interference and improve the SINR, the ZF detector can be utilized. In light of this discussion, in the following, we consider the time-reversal technique and aim at designing an additional ZF step to reduce the residual interference in TR-MRC.

We utilize the structure of OFDM to design a multiuser equalizer after time-reversal. In particular, we consider each subcarrier individually, and apply a zero-forcing matrix to eliminate the interference coming from different terminals. To pave the way for the development of a zero-forcing matrix, we consider a given subcarrier $p$, and reformulate (\ref{eqn:rbarTR_expanded}) as follows. Let the vector $\bar{\br}_{i}^{\tr}(p) = [\bar{r}^{\tr}_{0,i}(p),\dots,\bar{r}^{\tr}_{K-1,i}(p)]^{\rm T}$ contain the $p^{\rm th}$ output of the DFT blocks for different the terminals. Similarly, we define the noise vector $\bar{\bnu}_i^{\tr}(p) = [\bar{\nu}_{0,i}^{\tr}(p),\dots,\bar{\nu}_{K-1,i}^{\tr}(p)]^{\rm T}$. To express different interference terms, we construct the $K \times K$ matrices $\BG_{pq}^{(i,i-1)}$, $\BG_{pq}^{(i,i)}$ and $\BG_{pq}^{(i,i+1)}$ according to $\left[\BG_{pq}^{(i,i-1)}\right]_{kj} = \mathcal{G}_{kj,pq}^{(i,i-1)}$, $\left[\BG_{pq}^{(i,i)}\right]_{kj} = \mathcal{G}_{kj,pq}^{(i,i)}$, and $\left[\BG_{pq}^{(i,i+1)}\right]_{kj} = \mathcal{G}_{kj,pq}^{(i,i+1)}$. Following the above definitions, we can rearrange (\ref{eqn:rbarTR_expanded}) as 
\begin{align} \label{eqn:rbar_p} 
\bar{\br}_{i}^{\tr}(p) =  \BG_{pp}^{(i,i)} \bd_i(p) + \bxi_i(p) ,
\end{align}
where,
$
\bxi_i(p) = \sum_{\substack{q = 0 \\ q \neq p}}^{N-1}  \BG_{pq}^{(i,i)} \bd_i(q) + \sum_{q=0}^{N-1} \big( \BG_{pq}^{(i,i-1)} \bd_{i-1}(q) + \BG_{pq}^{(i,i+1)} \bd_{i+1}(q) \big) +  \bar{\bnu}_i^{\tr}(p) . \nonumber
$
We note that the term $\BG_{pp}^{(i,i)} \bd_i(p)$ contains the desired signals as well as the interference from symbols of different terminals transmitted in the same time/frequency slot as the time/frequency of interest, i.e, $i$ and $p$. More specifically, the diagonal elements of $\BG_{pp}^{(i,i)}$ correspond to the desired signal terms and the off-diagonal elements correspond to the interference terms. This interference is significant and is a source of performance degradation in multiuser scenarios. Hence, we propose to utilize the following ZF equalizer to remove the interference corresponding to the off-diagonal elements of $\BG_{pp}^{(i,i)}$.
\begin{align} \label{eqn:d_trzf}
\hat{\bd}_i(p) &= \big( \BG_{pp}^{(i,i)} \big)^{-1} \hspace{3pt} \bar{\br}_i^\tr(p) \nonumber \\
&= \bd_i(p) + \big( \BG_{pp}^{(i,i)} \big)^{-1} \hspace{3pt} {\bxi}_i(p) .
\end{align}
This additional equalization step leads to a substantial SINR performance improvement compared to the conventional TR-MRC. This is theoretically and numerically evaluated in Sections \ref{sec:tr_sinr} and \ref{sec:numerical_results}, respectively. Fig.~\ref{fig:Implementation} illustrates the baseband system implementation of the TR combining with the proposed ZF post-equalization.

\begin{figure*}[!t]
\centering
\includegraphics[scale=1.1]{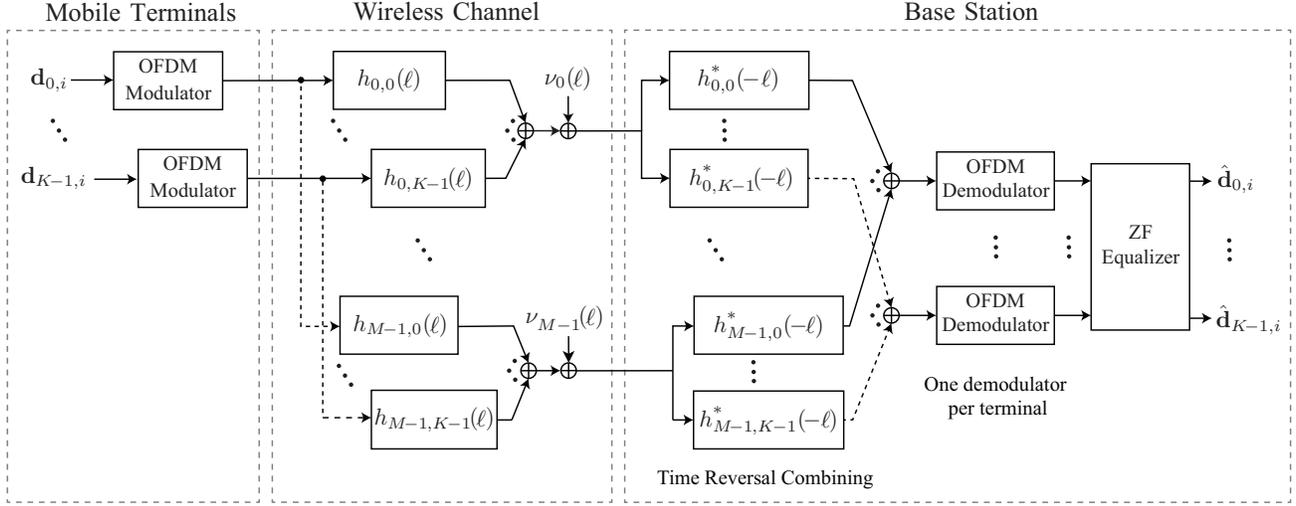} \vspace{-0.0cm}
\caption{Baseband system implementation of the proposed technique with TR-ZF receiver.}
\label{fig:Implementation}  \vspace{-0.0cm}
\end{figure*}

\section{Efficient Implementation and \\Complexity Analysis} \label{sec:implementation}

In this section, we study the computational complexity of the TR-MRC and TR-ZF receivers and compare the results with those of the conventional MRC and ZF methods utilized in CP-OFDM. The proposed receiver structures can be divided into two parts: (i) the TR part, and (ii) the post-equalization part. We discuss actions that should be taken to minimize the complexity of each part.

According to (\ref{eqn:rm2}), the TR part consists of a set of FIR filters whose complexity depends on the channel impulse responses between the BS antennas and the MTs. In particular, if the respective CIRs are sparse, i.e., are characterized by a small number of multipath components, one can directly implement the TR part in the time domain. However, in general, the direct implementation of (\ref{eqn:rm2}) may be computationally intensive in a wide-band OFDM transmission scenario, as the number of channel taps can be large.

Fortunately, the above issue can be resolved by utilizing the fast-convolution techniques such as \emph{overlap-add} and \emph{overlap-save},  \cite{oppenheim2010discrete}. Thus, the TR convolutions in (\ref{eqn:rm2}) are implemented efficiently in the frequency domain using the fast Fourier transform (FFT) algorithm. In the overlap-add and overlap-save methods, the processing is performed on a block-by-block basis, where each block is of length $\tilde{N}$ and is constructed from the samples of the input signal $r_m(l)$. Here, $\tilde{N}$ is a design parameter and is usually selected from the range $4L \leq \tilde{N} \leq 8L$ to minimize the computational cost. Accordingly, an $\tilde{N}$-point FFT is applied to each block to obtain the frequency-domain samples of the input signal. Then, these samples are multiplied with the respective frequency-domain channel coefficients. At this point, in order to minimize the number of required inverse FFT (IFFT) blocks, we can combine the signals corresponding to different BS antennas directly in the frequency domain, and then, apply a single IFFT block to the resulting signal to obtain the samples of $r_k^\tr(\ell)$. The above procedure significantly reduces the computational cost of the TR operation.

We now focus on the implementation of the second part, i.e., post-equalization. Direct calculation of the matrices involved in the ZF post-equalization introduced in Section \ref{sec:TRZF} imposes a substantial amount of computational burden to the system. In particular, considering a given subcarrier $p$, the matrix $\BG_{pp}^{(i,i)}$ should be computed to perform the ZF equalization according to (\ref{eqn:d_trzf}). The element $\left[\BG_{pp}^{(i,i)}\right]_{kj}$ is equal to the the $p^{\rm th}$ diagonal element of $\xbar{\bG}_{kj}^{(i,i)}$. Therefore, the direct approach requires the computation of the $p^{\rm th}$ diagonal elements of the matrices $\xbar{\bG}_{kj}^{(i,i)} = \bF_N \bG_{kj}^{(i,i)} \bF_N^{\rm H}$, for $k,j \in \{0,\dots,K-1\}$, to form the ZF post-equalization matrix. This involves a great number of calculations especially when the number of subcarriers is large. In particular, the number of complex multiplications using the direct method for all the subcarriers has a complexity that is of order $K^2N^3$. We denote this complexity by $\mathcal{O}(K^2N^3)$. Clearly, the direct method becomes computationally very expensive when $N$ is large. Fortunately, this issue can be resolved through the method that we introduce in the following.

\begin{table*}[!t] 
\renewcommand{\arraystretch}{1.4}
\caption{Computational Complexity of the Conventional MRC and ZF Detectors Utilized in CP-OFDM Systems.}
\label{tab:complexity1} \vspace{-0.0cm}
\centering
\begin{tabular}{|c|c|} \hline \hline
Technique & Number of Complex Multiplications \\
\hline \hline 
MRC    & $\frac{1}{2}QMN\log_2 N + QMNK$ \\ \hline
ZF     & $\frac{1}{2}QMN\log_2 N + QMNK + \frac{3}{2}MNK^2 + \frac{1}{3}NK^3$ \\ \hline \hline
\end{tabular} \vspace{-0.0cm}
\end{table*}

\begin{table*}[!t] 
\renewcommand{\arraystretch}{1.7}
\caption{Computational Complexity of Different Parts of the Receivers Proposed for OFDM without CP Systems.}
\label{tab:complexity2} \vspace{-0.0cm}
\centering
\begin{tabular}{|c|c|} \hline \hline
Technique & Number of Complex Multiplications \\
\hline \hline 
Time Reversal Combining  & $\frac{1}{2}\frac{QMN}{\tilde{N}-L+1} \tilde{N}\log_2 \tilde{N} + \frac{QMNK\tilde{N}}{\tilde{N}-L+1} + \frac{1}{2} \frac{QNK}{\tilde{N}-L+1} \tilde{N}\log_2 \tilde{N} + \frac{1}{2}QKN\log_2 N$ \\ \hline
ZF Post-Equalization  & $ \frac{K(K+1)}{2}M \tilde{N} + \frac{K(K+1)}{4}\tilde{N}\log_2 \tilde{N} + \frac{1}{2}K^2N\log_2 N + \frac{1}{3}NK^3 + QNK^2$  \\ \hline \hline
\end{tabular}  \vspace{-0.0cm}
\end{table*}

\begin{figure*}[!t] 
\centering
\subfigure[]{ \includegraphics[scale=0.57]{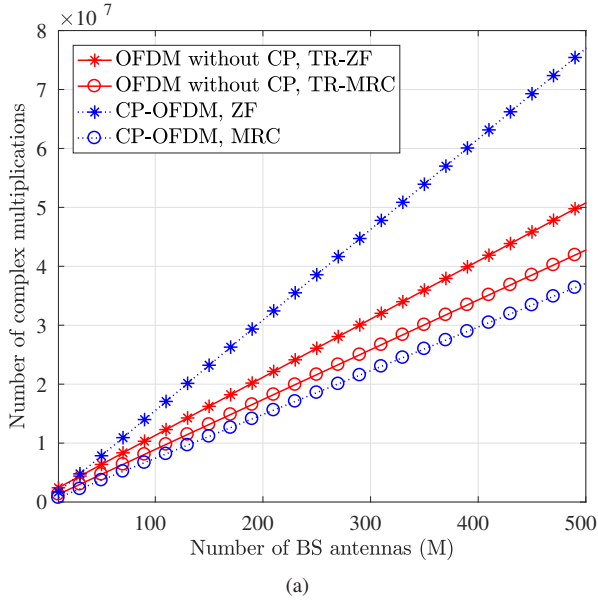}%
\label{fig:complexity1}} \hspace{0.7cm} %
\subfigure[]{ \includegraphics[scale=0.57]{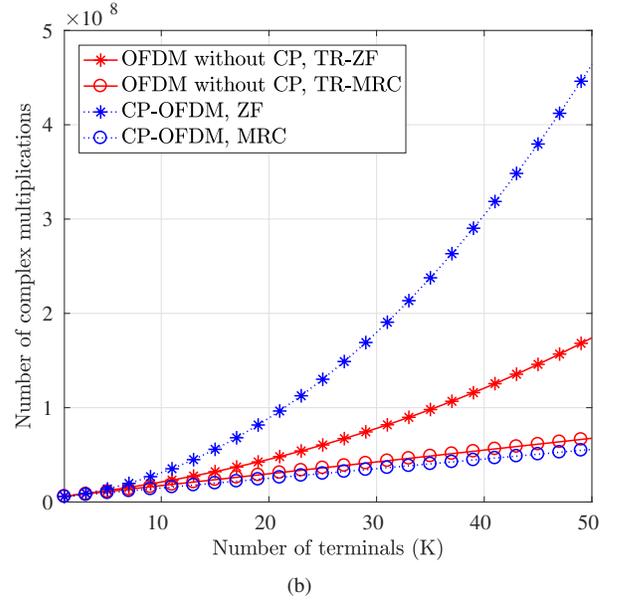}%
\label{fig:complexity2}}%
\caption{Computational complexity comparison of OFDM without CP with TR-MRC and TR-ZF techniques against CP-OFDM with MRC and ZF equalizations. Here, the following parameters are considered. $N=512$, $\tilde{N}=256$, $L=40$, and $Q=10$. In (a), $K=10$ is fixed and the value of $M$ is varied, whereas in (b), $M=200$ is fixed and the value of $K$ is varied.}
\label{fig:complexity}  \vspace{-0.0cm}
\end{figure*}

According to the expression $\xbar{\bG}_{kj}^{(i,i)} = \bF_N \bG_{kj}^{(i,i)} \bF_N^{\rm H}$, we have $\mathcal{G}_{kj,pp}^{(i,i)} = \bff_p^\mathrm{T} \bG_{kj}^{(i,i)} \bff_p^*$, where we recall that $\mathcal{G}_{kj,pp}^{(i,i)} \triangleq \left[ \xbar{\bG}_{kj}^{(i,i)} \right]_{pp}$. Therefore, one can obtain the element $\mathcal{G}_{kj,pp}^{(i,i)}$ as a linear combination of the samples of the equivalent CIR $g_{kj}(\ell)$, i.e., $\mathcal{G}_{kj,pp}^{(i,i)} = \sum_{\ell = -L+1}^{L-1}  \beta_\ell \hspace{2pt} g_{kj}(\ell)$, for some coefficients $\beta_\ell$. After some algebraic manipulations, the coefficients $\beta_\ell$ can be found as $\beta_\ell = \frac{N - |\ell |}{N} e^{-j \frac{2\pi \ell p}{N}}$. Hence, we have
\begin{align} \label{eqn:ZF_coeff_dft}
\mathcal{G}_{kj,pp}^{(i,i)} =  \frac{1}{N} \sum_{\ell = -L+1}^{L-1} \left(N - |\ell| \right) g_{kj}(\ell) e^{-j \frac{2\pi \ell p}{N} } .
\end{align}
Based on the above expression, $\mathcal{G}_{kj,pp}^{(i,i)}$ is equal to the $p^{\rm th}$ coefficient of the $N$-point DFT of the signal $g'_{kj}(\ell) \triangleq \frac{ N - |\ell| }{N}  g_{kj}(\ell) e^{j\frac{2\pi (L-1) p}{N} }$. Therefore, $\mathcal{G}_{kj,pp}^{(i,i)}$ can be computed efficiently using the FFT algorithm. Deploying this method, the number of complex multiplications needed for deriving the matrices involved in the ZF post-equalization is reduced to $\mathcal{O}\left(K^2N\log_2 N\right)$. As a result, a substantial computational complexity reduction is achieved.

We now compare the computational cost of TR-MRC and TR-ZF for OFDM without CP with those of the conventional MRC and ZF in CP-OFDM. In Table~\ref{tab:complexity1}, we have presented the number of complex multiplications needed to perform the MRC and ZF methods in CP-OFDM. Here, following our earlier notation, $Q$ represents the number of OFDM symbols. In Table ~\ref{tab:complexity1}, for both cases of MRC and ZF, the first and second terms represent the complexity due to the time-to-frequency conversion using $N$-point FFT blocks and frequency-domain combining, respectively. In the case of ZF, the third and fourth terms are due to the calculation of the ZF combining matrices $\BW_p = \xbar{\bH}_p (\xbar{\bH}_p^{\rm H} \xbar{\bH}_p)^{-1}$, $p\in \{0,\dots,N-1\}$. This needs to be calculated once for the transmitted packet consisting of $Q$ symbols.

Table~\ref{tab:complexity2} shows the number of complex multiplications needed to perform the TR combining and ZF post-equalization using the procedures discussed in this section. More specifically, the first three terms in the case of TR combining are due to the implementation of (\ref{eqn:rm2}) using fast-convolution as discussed above. Moreover, the fourth is arising from the calculation of $\bar{\br}^\tr_{k,i}$ from $\br^\tr_{k,i}$ using $N$-point FFT blocks. In the case of ZF post-equalization, the first two terms given in Table~\ref{tab:complexity2} account for the calculation of the equivalent channel responses $g_{kj}(\ell)$ given in (\ref{eqn:gkj}) using fast-convolution\footnote{Here, FFT size of $\tilde{N}$ is considered. Moreover, we have used the fact that $g_{jk}(\ell) = g_{kj}^*(-\ell)$ to reduce the number of computations.}.  The third term, i.e., $\frac{1}{2}K^2N\log_2 N$, is arising from the calculation of the coefficients $\mathcal{G}_{kj,pp}^{(i,i)}$ according to (\ref{eqn:ZF_coeff_dft}). The fourth term is due to the matrix inversion $\big( \BG_{pp}^{(i,i)} \big)^{-1}$. Finally, the last term accounts for the multiplication of the ZF equalization matrix to the input vector as in (\ref{eqn:d_trzf}).

Fig.~\ref{fig:complexity} compares the computational complexity of OFDM without CP with TR-MRC and TR-ZF techniques against CP-OFDM with MRC and ZF methods. Here, the following parameters are considered. $N=512$, $\tilde{N}=256$, $L=40$, and $Q=10$. In Fig.~\ref{fig:complexity1}, we have fixed $K=10$ and varied the value of $M$, whereas in Fig.~\ref{fig:complexity2}, $M=200$ is fixed and the value of $K$ is varied. As the figures show, while the complexity of MRC and TR-MRC are approximately the same, the TR-ZF receiver has a significantly lower computational cost compared to the ZF receiver. The reason for this is that the proposed ZF post-equalization takes place after multi-antenna combining; see Fig.~\ref{fig:Implementation}. Hence, the number of input signals to the ZF post-equalizer is significantly reduced as compared to the case of conventional ZF equalizer.

\section{Analysis of SINR and Achievable Rate} \label{sec:tr_sinr}

In this section, we analyze the SINR performance of both TR-MRC and TR-ZF receivers. This SINR analysis will ultimately lead us to find a lower-bound for the achievable information rate of each equalization technique.

\subsection{TR-MRC}

According to (\ref{eqn:rbarTR_expanded}), the SINR of the TR-MRC receiver can be calculated as
\be
\sinr_{k,p}^\trmrc = \frac{P_{k,p}^{\tt Sig} }{ P_{k,p}^{\tt ICI} + P_{k,p}^{\tt ISI} + P_{k,p}^{\tt MUI} + P_{k,p}^{\tt Noise} }, \label{eqn:TRMRC-SINR1} 
\ee
where $P_{k,p}^{\tt Sig} = \mathbb{E} \Big\{ \big| \mathcal{G}_{kk,pp}^{(i,i)} \big|^2 \Big\}$, $P_{k,p}^{\tt ICI} = \mathbb{E} \Big\{ \sum\limits_{\substack{q=0 \\ q \neq p}}^{N-1} \big| {\mathcal{G}}_{kk,pq}^{(i,i)} \big|^2 \Big\}$, $P_{k,p}^{\tt ISI} = \mathbb{E} \Big\{ \sum\limits_{q=0}^{N-1} \big( \big| {\mathcal{G}}_{kk,pq}^{(i,i-1)} \big|^2 + \big| {\mathcal{G}}_{kk,pq}^{(i,i+1)} \big|^2 \big) \Big\}$,  $P_{k,p}^{\tt MUI} = \mathbb{E} \Big\{ \sum\limits_{\substack{j=0 \\ j \neq k}}^{K-1} \sum\limits_{q=0}^{N-1} \big( \big| {\mathcal{G}}_{kj,pq}^{(i,i-1)} \big|^2 + \big| {\mathcal{G}}_{kj,pq}^{(i,i)} \big|^2 + \big| {\mathcal{G}}_{kj,pq}^{(i,i+1)} \big|^2 \big) \Big\}$, and $P_{k,p}^{\tt Noise} = \mathbb{E} \Big\{ \big| \bar{\nu}_{k,i}^\tr(p) \big|^2 \Big\}$. Using the channel model introduced in Section \ref{sec:SystemModel} and after some straightforward calculations, the average noise power can be obtained as $P_{k,p}^{\tt Noise} = \sigma_\nu^2$. In the following, in order to simplify the above SINR expression, we aim to analyze the interference coefficients ${\mathcal{G}}_{kj,pq}^{(i,i-1)}$, ${\mathcal{G}}_{kj,pq}^{(i,i)}$, and ${\mathcal{G}}_{kj,pq}^{(i,i+1)}$. By utilizing the Toeplitz structure of the matrices ${\bG}_{pq}^{(i,i-1)}$, ${\bG}_{pq}^{(i,i)}$ and ${\bG}_{pq}^{(i,i+1)}$, one can obtain these coefficients through the following expressions.
\begin{alignat}{4} \label{eqn:G}
&{\mathcal{G}}_{kj,pq}^{(i,i-1)}  &&= \bff_p^\mathrm{T} \bG_{kj}^{(i,i-1)} \bff_q^* &&= \ba_{pq}^\mathrm{H} \bg_{kj}, \nonumber \\
&{\mathcal{G}}_{kj,pq}^{(i,i)}    &&= \bff_p^\mathrm{T} \bG_{kj}^{(i,i)} \bff_q^*   &&= \bb_{pq}^\mathrm{H} \bg_{kj}, \nonumber \\
&{\mathcal{G}}_{kj,pq}^{(i,i+1)}  &&= \bff_p^\mathrm{T} \bG_{kj}^{(i,i+1)} \bff_q^* &&= \bc_{pq}^\mathrm{H} \bg_{kj},
\end{alignat}
where $\bff_p$ is the $p^{\rm th}$ column of the $N$-point DFT matrix, $\bF_N$, and the vector $\bg_{kj} \triangleq [g_{kj}(-L+1),\dots,g_{kj}(L-1)]^{\rm T}$ contains the samples of the TR channel impulse response $g_{kj}(\ell)$. Also, the vectors $\ba_{pq}$, $\bb_{pq}$, and $\bc_{pq}$ are determined by
\begin{align} \label{eqn:abc}
&\ba_{pq} = \frac{1}{\sqrt{N}} \Toeplitz_{L'\times N}\big(\big[\bzero_{1\times N+L-1},\omega_q^{N-1},\ldots,\omega_q^{N-L+1}\big]^\mathrm{T}\big) \bff_p^{*}, \nonumber \\
&\bb_{pq} = \frac{1}{\sqrt{N}} \Toeplitz_{L'\times N}\big(\big[\bzero_{1\times L-1},\omega_q^{N-1},\ldots,\omega_q^{0},\bzero_{1\times L-1}\big]^\mathrm{T}\big) \bff_p^{*},  \nonumber \\
&\bc_{pq} = \frac{1}{\sqrt{N}} \Toeplitz_{L'\times N}\big(\big[\omega_q^{L-2},\ldots ,\omega_q^{0},\bzero_{1\times N+L-1}\big]^\mathrm{T}\big) \bff_p^{*}, 
\end{align}
respectively, where $\omega_q \triangleq e^{-j\frac{2\pi q}{N}}$ and $L' \triangleq 2L-1$ is the length of the vector $\bg_{kj}$. We note that while the TR channel impulse response $\bg_{kj}$ has a random nature, the vectors $\ba_{pq}$, $\bb_{pq}$, and $\bc_{pq}$ are deterministic. Accordingly, in (\ref{eqn:G}), we have separated the random and deterministic parts of the coefficients ${\mathcal{G}}_{kj,pq}^{(i,i-1)}$, ${\mathcal{G}}_{kj,pq}^{(i,i)}$, and ${\mathcal{G}}_{kj,pq}^{(i,i+1)}$. This will help us to find their statistics.

Following the definition of the time-reversal equivalent channel response $g_{kj}(\ell)$ given in (\ref{eqn:gkj}), the mean of the complex random vector $\bg_{kj}$ can be obtained as
\be
\mathbb{E} \big\{\bg_{kj}\big\} = \sqrt{M} \delta_{kj}  {\bdelta}_{L'},  \label{eqn:gkj_mean}
\ee
where $\bdelta_{L'} \triangleq \left[\bzero_{1\times(L-1)},1,\bzero_{1\times(L-1)}\right]^{\rm T}$, and $\delta_{kj}$ is the Kronecker delta function. Moreover, the covariance matrix of $\bg_{kj}$ is calculated according to
\begin{align} \label{eqn:gkj_cov}
&\mathbb{E} \Big\{\Big(\bg_{kj}-\mathbb{E} \big\{\bg_{kj}\big\}\Big) \Big(\bg_{kj}-\mathbb{E} \big\{\bg_{kj}\big\}\Big)^{\mathrm{H}}\Big\} = \bGamma, 
\end{align}
where $\bGamma \triangleq \mathrm{diag}\{ \tilde{\brho} \}$. The elements in the vector $\tilde{\brho}=[\tilde{\rho}(-L+1),\dots,\tilde{\rho}(L-1)]^{\rm T}$ are obtained by convolving $\rho(\ell)$ by its time-reversed version, i.e., $\tilde{\rho}(i) = \sum_{\ell = 0}^{L-1} \rho(\ell) \rho(\ell - i)$.

According to (\ref{eqn:G}), the SINR expression given in (\ref{eqn:TRMRC-SINR1}) can be written as
\be \label{eqn:TRMRC-SINR2} 
\sinr^\trmrc_{k,p} = \frac{ \mathbb{E} \big\{ Q_{k,p}^{\tt Sig} \big\} }{ \mathbb{E} \big\{ Q_{k,p}^{\tt Intf} \big\} + \sigma_\nu^2} , 
\ee
where $Q^{\tt Intf}_{k,p} \triangleq \bg_{kk}^\mathrm{H} \bPhi_p \bg_{kk} + \sum_{\substack{j=0 \\ j \neq k}}^{K-1} \bg_{kj}^\mathrm{H} \bPsi_p \bg_{kj}$  includes the interference power due to the ICI, ISI, and MUI components, and $Q_{k,p}^{\tt Sig} \triangleq \bg_{kk}^{\rm H} \bB_p \bg_{kk}$ is the desired signal power. Here, $\bB_p = \bb_{pp} \bb_{pp}^{\rm H}$ and the matrices $\bPsi_p$, and $\bPhi_p$ are defined according to
\bse
\begin{align} 
\bPsi_p &= \sum_{q=0}^{N-1} \left( \ba_{pq}\ba_{pq}^\mathrm{H} + \bb_{pq}\bb_{pq}^\mathrm{H} + \bc_{pq}\bc_{pq}^\mathrm{H} \right), \\
\bPhi_p &= \sum_{\substack{q=0 \\ q\neq p}}^{N-1} \bb_{pq}\bb_{pq}^\mathrm{H} +  \sum_{q=0}^{N-1} \left( \ba_{pq}\ba_{pq}^\mathrm{H} + \bc_{pq}\bc_{pq}^\mathrm{H} \right) \nonumber \\ &= \bPsi_p - \bB_p ,
 \label{eqn:psi_p} 
\end{align}
\ese
respectively. We note that $Q_{k,p}^{\tt Intf}$ is a summations of $K$ \emph{quadratic} terms in the complex random vectors $\bg_{kj}$, $j \in \{0,\dots,K-1 \}$. Similarly, $Q_{k,p}^{\tt Sig}$ is quadratic in the complex random vector $\bg_{kk}$. 

\begin{proposition} \label{prp:SINR-TRMRC} 
In the absence of CP and with TR-MRC equalization, the SINR can be calculated as
\be \label{eqn:TRMRC-SINR3}
\sinr_{k,p}^\trmrc  = \frac{M + \lambda}{ K - \lambda + \sigma_\nu^2}  ,
\ee
where $\lambda \triangleq \sum\limits_{\ell = -L+1}^{L-1} \left( 1- \frac{|\ell|}{N} \right)^2 \tilde{\rho}(\ell)$. 
\end{proposition} 
\begin{proof} 
According to (\ref{eqn:gkj_mean}) and (\ref{eqn:gkj_cov}), the mean value of the quadratic term $Q_{k,p}^{\tt Sig}$ can be obtained as, $\mathbb{E}\big\{ Q_{k,p}^{\tt Sig} \big\}  = M + \trace \left\{ \bGamma \bB_p \right\}$, \cite[p.~53]{mathai1992quadratic}, where we have used the fact that $\left[\bb_{pp}\right]_L = 1$. Similarly, by noting that $\left[\ba_{pq}\right]_L = \left[\bc_{pq}\right]_L = 0$ for any $p$ and $q$, and $\left[\bb_{pq}\right]_L = 0$ for $q \neq p$, we can find the mean of the quadratic expression $Q_{k,p}^{\tt Intf}$ as $\mathbb{E}\big\{ Q_{k,p}^{\tt Intf} \big\} = \trace \left\{ \bGamma \bPhi_p \right\} + (K-1) ~ \trace \left\{ \bGamma \bPsi_p \right\}$. To simplify this, we note that the diagonal elements of $\bPsi_p$ are all equal to one. Accordingly, $\trace \{{\bGamma} {\bPsi}_p \} = \trace \{{\bGamma} \}  = \sum_i \tilde{\rho}(i) = \sum_i \sum_\ell {\rho}(\ell) \rho(\ell - i) = 1$. Hence, $\mathbb{E}\big\{ Q_{k,p}^{\tt Intf} \big\} = K - \trace \left\{ \bGamma \bB_p \right\}$.  The value of $\trace \{ \bGamma \bB_p \}$ can be obtained as follows. From (\ref{eqn:abc}) we can find the elements of the vector $\bb_{pp}$ according to 
\begin{align}  
&\bb_{pp} = \nonumber \\
& \-\ \Big[ \frac{N-L+1}{N} e^{j \frac{2\pi}{N} (L-1) p},\frac{N-L+2}{N}e^{j \frac{2\pi}{N} (L-2) p},\dots,1, \nonumber \\ 
& \-\ \-\ \dots, \frac{N-L+2}{N}e^{j \frac{2\pi}{N} (2-L) p},\frac{N-L+1}{N} e^{j \frac{2\pi}{N} (1-L) p} \Big]^{\rm T}. \label{eqn:bpp}
\end{align}
Hence, $\lambda \triangleq \trace \{ \bGamma \bB_p \} = \sum_{\ell = -L+1}^{L-1} \left( 1- \frac{|\ell|}{N} \right)^2 \tilde{\rho}(\ell)$. This completes the proof.
\end{proof}

\begin{remark} 
The SINR gain of $\mathcal{O}(M)$ is achievable with TR-MRC and the SINR saturation problem is resolved.
\end{remark}

It is worth mentioning that the parameter $\lambda = \sum_{\ell = -L+1}^{L-1} \left( 1- \frac{|\ell|}{N} \right)^2 \tilde{\rho}(\ell)$ is a positive constant that depends on the channel PDP. Moreover, using $\sum_\ell \tilde{\rho}(\ell) = 1$, we can find that $\lambda$ is always less than or equal to one, i.e., $\lambda\leq 1$. When the channel length is much smaller than the symbol duration, i.e., $L \ll N$, we have $\left( 1- \frac{|\ell|}{N} \right)^2 \approx 1$ for $\ell \in \{-L+1,\cdots,L-1\}$. This leads to $\lambda \approx 1$. For a fixed channel PDP, as the symbol duration $N$ becomes smaller, the value of $\lambda$ decreases.

Using the result of the Proposition~\ref{prp:SINR-TRMRC}, a lower bound on the achievable information rate at the output of the TR-MRC equalizer can be obtained by considering the worst case uncorrelated additive noise. Assuming that terminals transmit Gaussian data symbols, it is proven in \cite{hassibi2003much} that the worst case uncorrelated noise is circularly symmetric Gaussian with the same variance as the effective additive noise. Accordingly, a lower bound on the achievable rate in the case of TR-MRC can be obtained as 
\be \label{eqn:rate_trmrc}
R_k^\trmrc =  \log_2 \bigg( 1 + \frac{M + \lambda}{K - \lambda + \sigma_\nu^2} \bigg) .
\ee
On the other hand, a lower bound on the achievable information rate of CP-OFDM transmission with MRC equalizer is given by, \cite{ngo2013energy,marzetta2016fundamentals},
\be \label{eqn:rate_cpofdm_mrc} 
R_k^{\substack{\cpofdm \\ \mrc}} = \frac{N}{N + L} \log_2 \left(1 + \frac{M-1}{K-1+\sigma^2_\nu} \right) ,
\ee
where the term $\frac{N}{N+L}$ represents the rate loss due to the CP overhead. In Section \ref{sec:numerical_results}, we numerically evaluate the rate given in (\ref{eqn:rate_trmrc}) and compare it against (\ref{eqn:rate_cpofdm_mrc}) as a benchmark.

Before we end our discussion in this section, we note that for large values of $M$ and $K$, we have $R_k^\trmrc \approx \log_2 \left( 1 + \frac{M}{K + \sigma_\nu^2} \right)$. This matches the achievable rate reported in \cite{Pitarokoilis2012} for the case of single-carrier transmission when TR-MRC is applied. This implies that when TR-MRC is utilized, and for large values of $M$ and $K$, the same information rate can be achieved either by the OFDM without CP or the single-carrier transmission.

\subsection{TR-ZF}

In the case of TR-ZF, the additional ZF equalization step removes a significant portion of the remaining interference after the TR operation. Here, we mathematically analyze the SINR and achievable rate performance of this scheme.

In order to find the SINR performance of the TR-ZF receiver, we focus on the ZF equalization matrix $\big( \BG_{pp}^{(i,i)} \big)^{-1}$. We note that $\mathcal{G}_{kj,pp}^{(i,i)}$ is the element $kj$ of the matrix $\BG_{pp}^{(i,i)}$. Moreover, according to (\ref{eqn:G}) and (\ref{eqn:gkj_mean}), $\mathcal{G}_{kj,pp}^{(i,i)}$ can be expressed as $\mathcal{G}_{kj,pp}^{(i,i)} = \sqrt{M} \delta_{kj} [\bb_{pp}]_{L} + \bb_{pp}^{\rm H} \tilde{\bg}_{kj}$, where $\tilde{\bg}_{kj} \triangleq \bg_{kj} - \mathbb{E}\{\bg_{kj}\}$. Furthermore, as calculated in (\ref{eqn:bpp}), the $L^{\rm th}$ entry of the vector $\bb_{pp}$ is equal to $[\bb_{pp}]_{L} = 1$. Based on the above analysis, we can express the matrix $\BG_{pp}^{(i,i)}$ as 
\be \label{eqn:G_delta}
\BG_{pp}^{(i,i)} = \sqrt{M}~ \eye_K + \boldsymbol{\Delta}_p ,
\ee
where the elements of the matrix $\boldsymbol{\Delta}_p$ can be obtained according to $\left[ \boldsymbol{\Delta}_p \right]_{kj} = \bb_{pp}^{\rm H} \tilde{\bg}_{kj}$. According to (\ref{eqn:G_delta}), as the number of BS antennas $M$ grows large, the matrix $\frac{1}{\sqrt{M}} \BG_{pp}^{(i,i)}$ converges almost surely to $\eye_K$. Hence, $\BG_{pp}^{(i,i)}$ is asymptotically well-conditioned, and its inverse $\big( \BG_{pp}^{(i,i)} \big)^{-1}$ tends to $\frac{1}{\sqrt{M}} \eye_K$ as $M$ grows large. Using this, the following proposition finds the asymptotic ($M\rightarrow \infty$) SINR in the case of TR-ZF.

\begin{proposition} \label{prp:SINR-TRZF}
In the absence of CP and with TR-ZF equalization, the SINR tends to
\be  \label{eqn:TRZF-SINR1}
\sinr_{k,p}^\trzf  = \frac{M}{K(1-\lambda) + \sigma_\nu^2},
\ee
as $M$ grows large. We recall that $\lambda \triangleq \sum\limits_{\ell = -L+1}^{L-1} \left( 1- \frac{|\ell|}{N} \right)^2 \tilde{\rho}(\ell) \leq 1$.
\end{proposition}
\begin{proof}
According to (\ref{eqn:G_delta}), the ZF equalization matrix $\big( \BG_{pp}^{(i,i)} \big)^{-1}$ tends to $\frac{1}{\sqrt{M}} \eye_K$ as the number of BS antennas $M$ grows large. Therefore, the second term in (\ref{eqn:d_trzf}) tends to $\frac{1}{\sqrt{M}} ~ \bxi_i(p)$ asymptotically. We note that this term constitutes the residual interference after the TR-ZF equalization. Using the same line of derivation as in Proposition \ref{prp:SINR-TRMRC}, we can find the variance of the elements in $\bxi_i(p)$ as $\sigma_{\xi}^2 = K ~ \trace \{ \bGamma \bPhi_p \} + \sigma_\nu^2 = K (1 - \lambda) + \sigma_\nu^2 $. This leads to the SINR expression given in (\ref{eqn:TRZF-SINR1}).
\end{proof}

\begin{remark} 
Similar to the case of TR-MRC, the SINR gain of $\mathcal{O}(M)$ is achievable using TR-ZF receiver and the SINR saturation is avoided.
\end{remark}

The above result suggests that SINR saturation can be avoided through utilization of TR. The additional ZF equalization further improves the SINR level in multi-user systems.

According to (\ref{eqn:TRZF-SINR1}), a lower bound on the asymptotic achievable information rate at the output of the TR-ZF equalizer can be obtained as 
\begin{align} \label{eqn:rate_trzf} 
\tilde{R}_k^\trzf = \log_2 \bigg( 1 + \frac{M}{K(1-\lambda) + \sigma_\nu^2} \bigg) ,
\end{align}
where the tilde sign in $\tilde{R}$ signifies that it is an asymptotic information rate, i.e., it tends to the actual information rate as the number of BS antennas $M$ increases.
On the other hand, the achievable information rate of CP-OFDM transmission with ZF equalizer is given by, \cite{ngo2013energy,marzetta2016fundamentals},
\be \label{eqn:rate_cpofdm_zf} 
R_k^{\substack{\cpofdm \\ \zf}} = \frac{N}{N + L} \log_2 \left(1 + \frac{M - K}{\sigma^2_\nu} \right) ,
\ee 
where the term $\frac{N}{N+L}$ represents the rate loss due to the CP overhead. We note that comparing (\ref{eqn:rate_trzf}) and (\ref{eqn:rate_cpofdm_zf}) may not be fair as the former is derived using asymptotic analysis, and the latter is valid for finite values of $M$ as well. Hence, for the purpose of comparison, we also consider the asymptotic version of (\ref{eqn:rate_cpofdm_zf}) given by, \cite{ngo2013energy},
\be \label{eqn:rate_cpofdm_zf_asymp} 
\tilde{R}_k^{\substack{\cpofdm \\ \zf}} = \frac{N}{N + L} \log_2 \left(1 + \frac{M}{\sigma^2_\nu} \right) .
\ee 
In Section \ref{sec:numerical_results}, we numerically evaluate the rate given in (\ref{eqn:rate_trzf}) and compare it against (\ref{eqn:rate_cpofdm_zf_asymp}) as a benchmark.

\section{Numerical Results} \label{sec:numerical_results}

In this section, we evaluate the analyses and discussions of the previous sections through numerical simulations. We consider the Extended Typical Urban (ETU) channel model as defined in the long term evolution (LTE) standard, \cite{3gpp.36.101}. We adopt the LTE air interface parameters to OFDM without CP. Specifically, the OFDM useful symbol duration of $T = 66.7~\mu$s, which translates to the subcarrier spacing of $\Delta f = 15$ kHz is considered. Note that when considering OFDM without CP transmission, the useful symbol duration is equal to the total symbol duration, and delay spread of the ETU model covers about $7\%$ of the OFDM symbol duration. We choose the DFT size of $N=512$, and $300$ active subcarriers. This corresponds to the $5$ MHz bandwidth scenario defined in the LTE standard.

\begin{figure}[!t]
\centering
\includegraphics[scale=0.58]{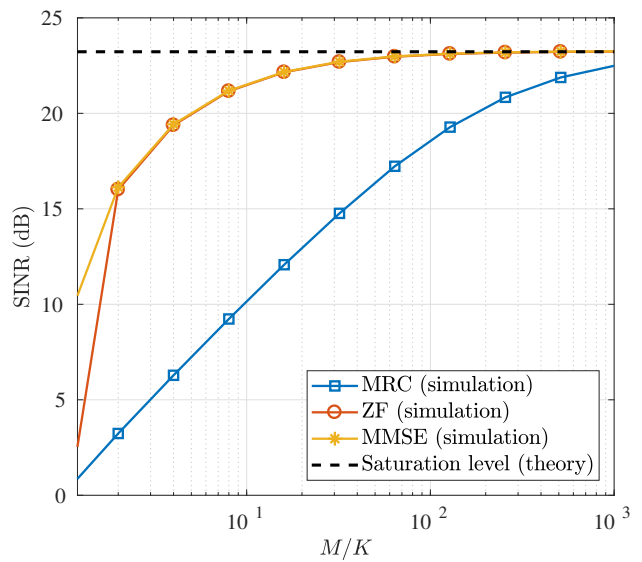} \vspace{-0.0cm}
\caption{SINR saturation in the case of conventional frequency-domain combiners. Here, $K=10$ terminals are considered and the number of BS antennas is varied. The SNR level is chosen to be $10$ dB. The saturation level is calculated using (\ref{eqn:SIR_sat}).}
\label{fig:sinr_comp_fc}
\end{figure}

\begin{figure}[!t]
\centering
\includegraphics[scale=0.58]{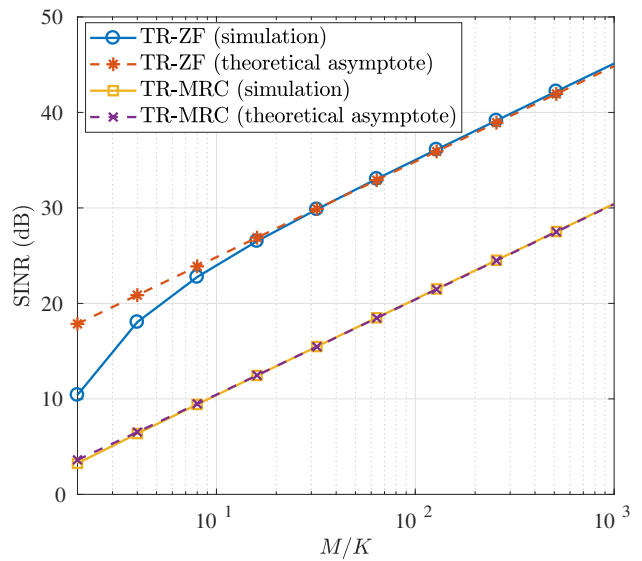} \vspace{-0.0cm}
\caption{SINR performance comparison for time reversal methods. Here, $K=10$ terminals are considered and the number of BS antennas is varied. The SNR level is chosen to be $10$ dB. Asymptotic theoretical SINR values are calculated according to (\ref{eqn:TRMRC-SINR3}) and (\ref{eqn:TRZF-SINR1}) for the cases of TR-MRC and TR-ZF, respectively. Using time reversal, arbitrarily large SINR values can be achieved by increasing the number of BS antennas.}
\label{fig:sinr_comp_tr}  \vspace{-0.0cm}
\end{figure}

\begin{figure*}[!t]
\centering
\subfigure[]{ \includegraphics[scale=0.58]{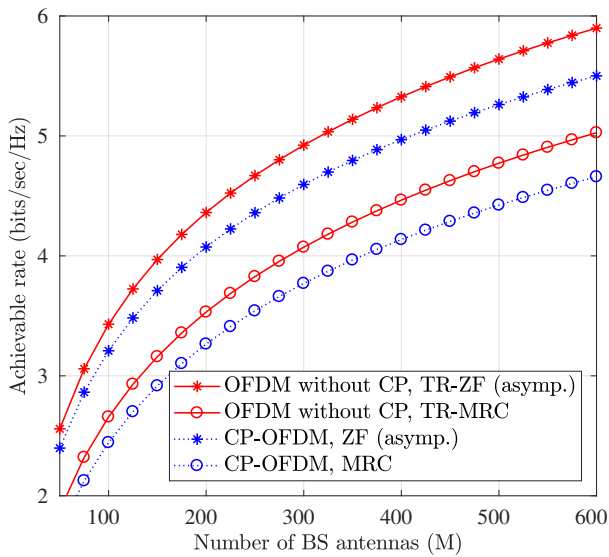}%
\label{fig:rate_M_K10}} \hspace{0.7cm}%
\subfigure[]{ \includegraphics[scale=0.58]{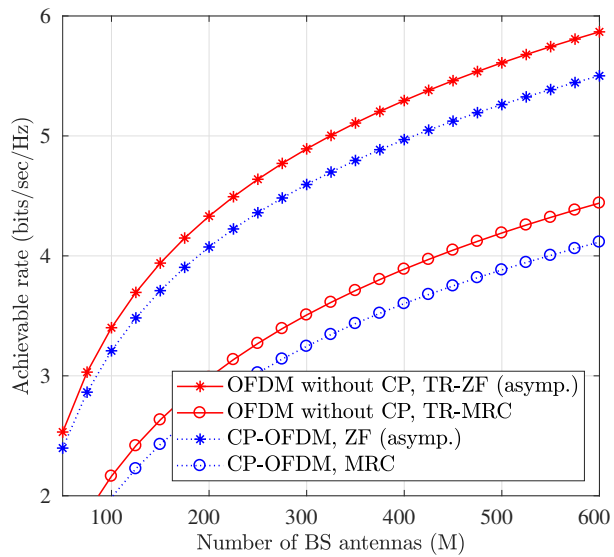}%
\label{fig:rate_M_K20}}%
\caption{Per user achievable information rate with and without the CP overhead. Here, the ratio $L/N$ is approximately 7\%, and the SNR level is chosen to be $-10$ dB. (a) $K=10$, (b) $K=20$ user terminals.}
\label{fig:rate_M}  \vspace{-0.0cm}
\end{figure*}

\begin{figure}[!t]
\centering
\includegraphics[scale=0.58]{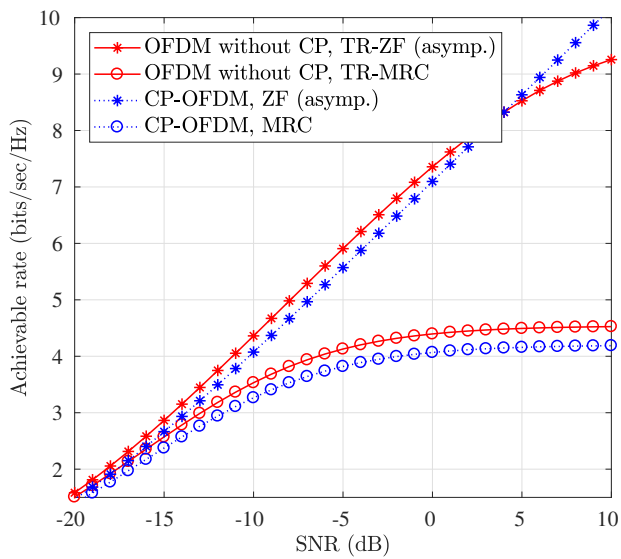} \vspace{-0.0cm}
\caption{Per user achievable information rate as a function of the SNR level. Here, the ratio $L/N$ is approximately 7\%, and $M=200$ BS antennas and $K=10$ terminals are considered.}
\label{fig:rate_SNR} \vspace{-0.cm}
\end{figure}

\begin{figure}[!t]
\centering 
\includegraphics[scale=0.58]{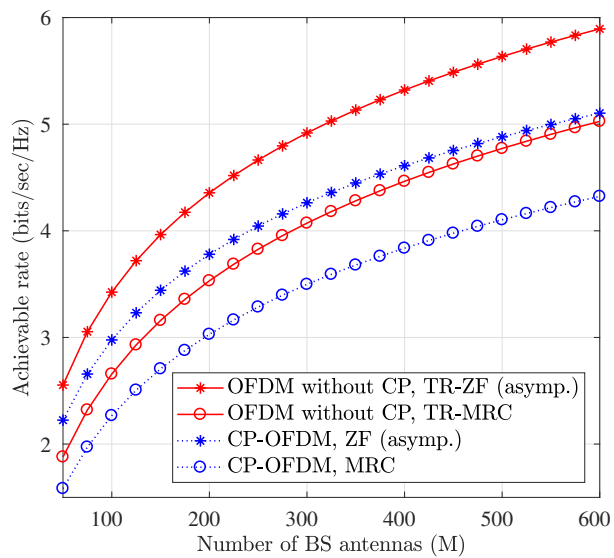} \vspace{-0.0cm}
\caption{Per user achievable information rate with and without the CP overhead. Here, $K = 10$ terminals are considered and the SNR level is chosen to be $-10$ dB. Moreover, the TDL-A channel with the RMS delay spread of $1100$~ns is assumed. In this channel model, the ratio $L/N$ is approximately 15\%.}
\label{fig:rate_M_TDL} \vspace{-0.0cm}
\end{figure}

We first evaluate the SINR performance of various methods discussed in this paper. In Fig.~\ref{fig:sinr_comp_fc}, we have demonstrated the SINR saturation of the conventional frequency-domain combining methods, namely MRC, ZF, and MMSE. In this experiment, $K=10$ active terminals are considered, and the noise level is chosen such that the average SNR at the input of the BS antennas is $10$ dB. We show the average SINR values over different channel realizations with the power delay profile of the ETU channel model. The saturation level is calculated using (\ref{eqn:SIR_sat}) and is compared with the simulated SINR values. As we expect, in all three frequency-domain combining methods, SINR does not improve beyond a certain deterministic level. As mentioned in Section \ref{sec:TR}, this problem can be resolved by using the TR technique. Fig.~\ref{fig:sinr_comp_tr} shows the SINR performance of TR-MRC and TR-ZF methods. Again, as expected, for both cases of TR-MRC and TR-ZF, SINR will grow unboundedly as the number of BS antennas grows. Moreover, since the proposed TR-ZF method significantly reduces the MUI level compared to the conventional TR-MRC technique, it yields to an improved SINR performance. In Fig.~\ref{fig:sinr_comp_tr},  the SNR at the input of the BS antennas is 10 dB. Moreover, we have also shown the theoretical SINR values calculated according to (\ref{eqn:TRMRC-SINR3}) and (\ref{eqn:TRZF-SINR1}) for the cases of TR-MRC and TR-ZF, respectively. As the number of BS antennas $M$ grows large, the simulated SINR values coincide with the values derived using asymptotic analysis in Section \ref{sec:tr_sinr}.

We next conduct an experiment to evaluate the achievable information rate with and without including the CP overhead. Fig.~\ref{fig:rate_M_K10} shows the theoretical achievable rate of OFDM without CP with TR-MRC and TR-ZF equalizers as well as that of CP-OFDM with MRC and ZF detectors. In the cases of OFDM without CP with TR-ZF and CP-OFDM with ZF equalizer, asymptotic rates given by (\ref{eqn:rate_trzf}) and (\ref{eqn:rate_cpofdm_zf_asymp}), respectively,  are considered.
In this experiment, $K=10$ terminals are considered and the noise level is chosen such that SNR at the input of the BS antennas is $-10$ dB. Fig.~\ref{fig:rate_M_K20} shows the results for the case where $K=20$ terminals are active. As shown in Figs.~\ref{fig:rate_M_K10} and \ref{fig:rate_M_K20}, with OFDM without CP and TR-MRC equalization, we can achieve a higher spectral efficiency as compared to in CP-OFDM with MRC equalizer. A similar argument applies for OFDM without CP with TR-ZF and CP-OFDM with ZF detector. Hence, as expected, by eliminating the CP overhead we can achieve a higher spectral efficiency compared with the conventional CP-OFDM systems. It should be noted that according to Figs.~\ref{fig:rate_M_K10} and \ref{fig:rate_M_K20}, for a fixed achievable rate performance, one can decrease the number of BS antennas (and hence the implementation cost) by removing the CP overhead.

In Fig.~\ref{fig:rate_SNR}, we compare the achievable rate performance of OFDM without CP and CP-OFDM for various levels of SNR. In this experiment, $M=100$ BS antennas and $K=10$ terminals are considered. As shown, for typical SNR levels, higher spectral efficiency can be achieved using OFDM without CP. On the other hand, in very low SNR regime, the noise level dominates the overall interference plus noise, and hence, similar rates can be achieved using OFDM with/without CP deploying various equalization methods. On the other hand, when the SNR level is high, the residual interference dominates the noise, hence the performance of OFDM without CP with TR-MRC/TR-ZF and CP-OFDM with MRC becomes saturated and does not improve with increasing the transmission power.

So far in this section, we considered the ETU channel model, which covers about $7\%$ of the OFDM symbol duration of $T = 66.7~\mu$s. In the next experiment, we aim to show the advantage of the elimination of CP in channels with larger delay spreads. Accordingly, we consider the TDL-A channel PDP with the RMS (root mean square) delay spread of $1100$~ns. This channel model has been recently proposed for the frequency spectrum above 6~GHz \cite{3gpp.38.900}, and covers about $15\%$ of the OFDM symbol duration. Fig.~\ref{fig:rate_M_TDL} shows the achievable rate comparison of OFDM without CP and CP-OFDM considering the above channel model. Here, $K=10$ terminals and the SNR level of $\text{SNR} = -10$~dB are considered. As shown, here due to a larger CP duration, the spectral efficiency is improved more considerably by eliminating the CP overhead.

\section{Conclusion} \label{sec:conclusion}

It is known that in massive MIMO channels uncorrelated noise and multiuser interference vanish as the number of BS antennas grows large. Motivated by this, in this paper, we studied OFDM without CP under such channels to investigate if the channel distortions (i.e., ISI and ICI) average out in the large antenna regime. To this end, we mathematically analyzed the asymptotic SINR performance of the conventional frequency-domain combining methods, i.e., MRC, ZF, and MMSE. Our analysis revealed that in these cases, there always exists some residual interference even for an infinite number of BS antennas leading to the saturation of the SINR performance. To solve this saturation issue, we proposed to use the TR technique. Moreover, we introduced a ZF equalization to be incorporated after the TR combining to further reduce the multiuser interference. We mathematically analyzed the asymptotic achievable information rate of the proposed receiver design. We showed that by removing the CP overhead and using the proposed technique, a higher spectral efficiency is achievable as compared to the conventional CP-OFDM systems, while the computational complexity is also reduced.

\appendix[Proof of the Results in (\ref{eqn:lambdas_simplified})]
\section{Proof of the Results in (\ref{eqn:lambdas_simplified})} \label{sec:appendix_freq}

The elements of $\xbar{\bH}_{m,k}^{(i,i-1)}$ and $\xbar{\bH}_{m,k}^{(i,i)}$ can be expanded as \cite{Molisch2007}
\be
\left[\xbar{\bH}_{m,k}^{(i,i-1)}\right]_{pq} \hspace{-0.15cm} = \hspace{-0.1cm}\frac{1}{N}\hspace{-0.1cm}\sum^{N-1}_{n=0}\sum^{L-1}_{\ell=0}\hspace{-0.05cm}h_{m,k}(\ell)e^{j\frac{2\pi}{N}(nq-\ell q-np)}w(n-\ell+N), \nonumber
\ee 
and
\be
\left[\xbar{\bH}_{m,k}^{(i,i)}\right]_{pq} = \frac{1}{N}\sum^{N-1}_{n=0}\sum^{L-1}_{\ell=0}h_{m,k}(\ell)e^{j\frac{2\pi}{N}(nq-\ell q-np)}w(n-\ell), \nonumber
\ee
where $w(n)$ is the windowing function, which is considered to be a rectangular window, i.e., $w(n)=
\bigg\{
	\begin{array}{ll}
	1, &0\leqslant n \leqslant N-1, \\
	0, & \rm{otherwise}.
	\end{array}$. Accordingly, following (\ref{eqn:lambda_asym}), we have
\begin{align}
& \mathcal{H}_{kk,pp}^{(i,i)} \rightarrow \mathbb{E}\bigg\{ \bar{h}_{m,k}^\ast(p) \left[ \xbar{\bH}^{(i,i)}_{m,k}\right]_{pp}\bigg\} \nonumber \\
& = \frac{1}{N}\mathbb{E} \bigg\{ \sum^{N-1}_{n=0}\sum^{L-1}_{\ell=0}\bar{h}_{m,k}^\ast(p) h_{m,k}(\ell)e^{-j\frac{2\pi}{N}\ell p}w(n-\ell) \bigg\} \nonumber \\
& = \frac{1}{N}\mathbb{E} \bigg\{ \sum^{N-1}_{n=0}\sum^{L-1}_{\ell=0}\sum^{L-1}_{\ell^{\prime}=0}{h}_{m,k}^\ast (\ell^{\prime})h_{m,k}(\ell)e^{j\frac{2\pi}{N}(\ell^{\prime}-\ell) p} w(n-\ell) \bigg\} \nonumber \\
& = \frac{1}{N}\sum^{L-1}_{\ell=0}(N-\ell)\rho(\ell) = 1-\frac{{\tau_{\rm av}}}{N}, \nonumber
\end{align}
and for $p\neq q$ we have,
\begin{align}
& \mathcal{H}_{kk,pq}^{(i,i)} \rightarrow \mathbb{E}\bigg\{ \bar{h}_{m,k}^\ast(p) \left[ \xbar{\bH}^{(i,i)}_{m,k}\right]_{pq}\bigg\} \nonumber \\
& = \frac{1}{N}\mathbb{E} \bigg\{ \sum^{N-1}_{n=0}\sum^{L-1}_{\ell=0}\sum^{L-1}_{\ell^{\prime}=0}{h}_{m,k}^\ast (\ell^{\prime})h_{m,k}(\ell) e^{j\frac{2\pi}{N}(nq - \ell q + \ell^{\prime}p - np)}  \nonumber \\
& \hspace{40pt} \times w(n-\ell) \bigg\} \nonumber \\
& = \frac{1}{N}\sum^{N-1}_{n=0}\sum^{L-1}_{\ell=0}\rho(\ell)e^{-j\frac{2\pi}{N}(\ell-n)(q-p)}w(n-\ell) \nonumber \\
& = \frac{1}{N}\sum^{L-1}_{\ell=0}\rho(\ell)e^{-j\frac{2\pi\ell (q-p)}{N}}\sum^{N-1}_{n=\ell}e^{j\frac{2\pi n(q-p)}{N}} \nonumber \\
& = -\frac{1}{N}\sum^{L-1}_{\ell=0}\rho(\ell)e^{-j\frac{2\pi\ell (q-p)}{N}}\frac{1-e^{j\frac{2\pi\ell (q-p)}{N}}}{1-e^{j\frac{2\pi (q-p)}{N}}} \nonumber \\
& = \frac{-1}{N(1-e^{j\frac{2\pi (q-p)}{N}})}\bigg(\sum^{L-1}_{\ell=0}\rho(\ell)e^{-j\frac{2\pi\ell (q-p)}{N}} -\sum^{L-1}_{\ell=0}\rho(\ell)\bigg) \nonumber \\
& = \frac{1-\bar{\rho}(q-p)}{N(1-e^{j\frac{2\pi (q-p)}{N}})}, \nonumber
\end{align}
where $\bar{\rho}(q) \triangleq \sum^{L-1}_{\ell=0}\rho(\ell)e^{-j\frac{2\pi\ell q}{N}}$. Similarly, the asymptotic value of the ISI coefficient $\mathcal{H}^{(i,i-1)}_{kk,pq}$ can be calculated as $\mathcal{H}_{kk,pp}^{(i,i-1)} \rightarrow \frac{{\tau_{\rm av}}}{N}$ and $\mathcal{H}_{kk,pq}^{(i,i-1)} \rightarrow \frac{\bar{\rho}(q-p)-1}{N(1-e^{j\frac{2\pi (q-p)}{N}})}$, when $p \neq q$. Moreover, with similar derivations it is possible to show that $\bar{h}_{m,k}(p)$ is uncorrelated with $\left[ \xbar{\bH}^{(i,i)}_{m,j}\right]_{pq}$ and $\left[ \xbar{\bH}^{(i,i-1)}_{m,j}\right]_{pq}$, when $k\neq j$. Accordingly, the MUI coefficients $\mathcal{H}^{(i,i)}_{kj,pq}$ and $\mathcal{H}^{(i,i-1)}_{kj,pq}$ tend to be zero as $M$ grows large.

\bibliographystyle{IEEEtran} 
\bibliography{IEEEabrv,MassiveMIMO}

\end{document}